\documentclass[journal]{IEEEtran}

\usepackage[fleqn]{amsmath}
\usepackage{amssymb,amsthm,graphicx,subfig,verbatim,enumitem,cite}
\usepackage{color}

\newtheorem{theorem}{Theorem}

\newtheorem{lemma}{Lemma}
\newtheorem{corollary}{Corollary}
\newtheorem{remark}{Remark}

\providecommand{\norm}[1]{|#1|}

\renewcommand\qed{$\blacksquare$}

\graphicspath{{figures/}}

\begin{document}

\title{Decentralized Event-Triggering for Control of Nonlinear Systems}

\author{Pavankumar Tallapragada and Nikhil Chopra
\thanks{This work was partially supported by the by the Office of Naval Research under grant number  N000141310160 and by National Science Foundation through grant number 1232127.}
\thanks{Pavankumar Tallapragada is with the Department of Mechanical Engineering,
        University of Maryland, College Park, 20742 MD, USA
        {\tt\small pavant@umd.edu}}%
\thanks{N. Chopra is with the Department of Mechanical Engineering and The Institute for Systems Research,
        University of Maryland, College Park, 20742 MD, USA
        {\tt\small nchopra@umd.edu}}%
}

\maketitle
\begin{abstract}
This paper considers nonlinear systems with full state feedback, a central controller and distributed sensors not co-located with the central controller. We present a methodology for designing decentralized asynchronous event-triggers, which utilize only locally available information, for determining the time instants of transmission from the sensors to the central controller. The proposed design guarantees a positive lower bound for the inter-transmission times of each sensor, while ensuring asymptotic stability of the origin of the system with an arbitrary, but priorly fixed, compact region of attraction. In the special case of Linear Time Invariant (LTI) systems, global asymptotic stability is guaranteed and scale invariance of inter-transmission times is preserved. A modified design method is also proposed for nonlinear systems, with the addition of event-triggered communication from the controller to the sensors, that promises to significantly increase the average sensor inter-transmission times compared to the case where the controller does not transmit data to the sensors. The proposed designs are illustrated through simulations of a linear and a nonlinear example.
\end{abstract}

\section{Introduction}

State based aperiodic event-triggering is receiving increased attention (a representative list of the recent literature includes \cite{tabuada2007, heemels2008, astrom2008, velasco2008, wang2010, lunze2010, lemmon2011survey, heemels2012intro}) as an alternative to the traditional time-triggering (example: periodic triggering) in sampled data control systems. In event based control systems, a state or data dependent event-triggering condition implicitly determines time instances at which control is updated or when a sensor transmits data to a controller. Such updates or transmissions are in general aperiodic and depend on the system state. Such a paradigm is particularly appealing in control systems with limited computational and/or communication resources.

Much of the literature on event-triggered control utilizes the full state information in the triggering conditions. However, in two very important classes of problems full state information is not available to the event-triggers. These are systems with decentralized sensing and/or dynamic output feedback control. In the latter case, full state information is not available even when the sensors and the controller are centralized (\textit{co-located}). In systems with decentralized sensing, each individual sensor has to base its decision to transmit data to a central controller only on locally available information. These two classes of problems are receiving attention in the community only recently - \cite{mazo2011, mazo2011async, mazo2012, mazo2012arxiv, wang2009dist, wang2011, depersis2013, pavan2013MSC} (decentralized sensing) and \cite{donkers2010, pavan2012necsys, pavan2012cdc, pavan2013cdc, lehmann2011, li2011, almeida2012} (output feedback control). This paper is an important addition to the limited literature on decentralized event-triggering in control systems with distributed sensors.

The basic contribution of this paper is a methodology for designing implicitly verified decentralized event-triggers for control of nonlinear systems. The system architecture we consider is one with full state feedback but with the sensors distributed and not co-located with a central controller. The proposed design methodology provides event-triggers that determine when each sensor transmits data to a central controller. The event-triggers are designed to utilize only locally available information, making the transmissions from the sensors asynchronous. The proposed design guarantees asymptotic stability of the origin of the system with an arbitrary, but fixed \textit{a priori}, compact region of attraction. It also guarantees a positive lower bound for the inter-transmission times of each sensor individually. In the special case of Linear Time Invariant (LTI) systems, global asymptotic stability is guaranteed and scale invariance of inter-transmission times is preserved. For nonlinear systems, we also propose a variant with the addition of event-triggered communication from the central controller to the sensors. This variant shows promise of significantly increasing the ``average" sensor inter-transmission times compared to the case when the controller does not transmit data to the sensors. Although our results fall short of mathematically proving an increase in the ``average" sensor inter-transmission times, some suggestive arguments are provided.  Nevertheless, results are provided that guarantee or suggest design choices that would increase lower bounds on inter-transmission times.

In the literature, distributed event-triggered control was studied with the assumption of weakly coupled subsystems in \cite{wang2009dist, wang2011} (finite $\mathcal{L}_p$ stable systems) and in \cite{depersis2013} (ISS and iISS systems), which allowed the design of event-triggers depending on only local information. While we do not make an assumption on the coupling strength, we do assume (in the nonlinear case) the knowledge of a global ISS Lyapunov function - an assumption that is sometimes seen as restrictive. However, note that we make the assumption that the ISS Lyapunov function is defined globally purely for ease of exposition. Since our results for the nonlinear systems guarantee only semi-global asymptotic stability, it is sufficient to have a semi-global ISS Lyapunov function, a much less severe requirement and one which may possibly be derived from Lyapunov functions that guarantee asymptotic stability (see \cite{pavan2013tac} for example). Further, with the event-triggers of \cite{wang2009dist, wang2011}, no positive lower bound for inter-transmission times exists in the global or even semi-global sense (see Remark \ref{rem:wang_compare} in the sequel).

Our proposed scheme is more closely related, at least in its aims and ideas, with \cite{mazo2011async, mazo2012, mazo2012arxiv} and with \cite{mazo2011} to a lesser extent. Our proposed scheme shares a similarity with \cite{mazo2011} in the use of timers (or a dwell time) in the event-triggers. One may view such inclusion of an implicitly verified dwell time as a combination of time-triggering and event-triggering. In the scheme proposed in \cite{mazo2011}, decentralized asynchronous event-triggering at the sensors is only utilized to request the central controller to seek a \textit{synchronous} snapshot of all the sensors' data. In order to reject closely timed requests by the sensors, the controller utilizes a timer whose expiration time threshold is determined from the properties of the corresponding centralized event-triggered control system. Our proposed event-triggers also include similar timers - however, their use, location and the computation of the expiration time thresholds are qualitatively different. One could say that the proposed scheme complements \cite{mazo2011} with the use of timers in each of the event-triggers at the sensors and as a result allows the controller to depend only on the asynchronously transmitted data. These differences would be apparent in the course of the paper and in particular in Remark \ref{rem:compare_mazo2011}. \cite{mazo2011async, mazo2012, mazo2012arxiv} proposed an interesting scheme for decentralized and asynchronous event-triggering with fixed and time-varying thresholds. However, the scheme guarantees only semi-global practical stability even for linear systems if the sensors do not to listen to the central controller. To the best of our knowledge, our work is unique in that our decentralized event-triggering scheme requires only unidirectional communication from the sensors to the controller and still guarantees semi-global asymptotic stability (global for LTI systems). As a result, no receivers are required at the sensor nodes - a feature that is very useful in low energy applications.

In the dynamic output feedback control literature, \cite{donkers2010, pavan2012necsys, pavan2012cdc, pavan2013cdc} consider asynchronous and decentralized event-triggering for Linear Time Invariant (LTI) systems. Again, the method in \cite{donkers2010} can guarantee only semi-global practical stability. In \cite{pavan2012necsys, pavan2012cdc, pavan2013cdc}, we have proposed a method that guarantees global asymptotic stability and positive minimum inter-transmission times. In fact, the basic principle in \cite{pavan2012necsys, pavan2012cdc, pavan2013cdc} is also used in this paper, with additional considerations for the nonlinear systems. The portion on the linear systems in this paper was presented in \cite{pavan2013MSC}.

The rest of the paper is organized as follows. Section \ref{sec:prob_setup} describes and formally sets up the problem under consideration. In Section \ref{sec:dec_async_event_trig}, the design of asynchronous decentralized event-triggers for nonlinear systems is presented - without, and then with, feedback from the central controller. Section \ref{sec:lin_sys} presents the special case of Linear Time Invariant (LTI) systems. The proposed design methodology is illustrated through simulations in Section \ref{sec:sim_results} and finally Section \ref{sec:conc} provides some concluding remarks.

\section{Problem Setup}
\label{sec:prob_setup}

Consider a nonlinear control system
\begin{align}
\dot{x} &= f(x,u), \quad x \in \mathbb{R}^n, \,\, u \in \mathbb{R}^m \label{eqn:nonlin_sys}
\end{align}
with the feedback control law
\begin{equation}
u = k(x+x_e) \label{eqn:control}
\end{equation}
where $x_e$ is the error in the measurement of $x$. In general, the measurement error can be due to many factors such as sensor noise and quantization. However, we consider measurement error that is purely a result of ``sampling" of the sensor data $x$. Before going into the precise definition of this measurement error, we first describe the broader problem. First, let us express \eqref{eqn:nonlin_sys} as a collection of $n$ scalar differential equations
\begin{equation}
\dot{x}_i = f_i(x,u), \quad x_i \in \mathbb{R}, \quad i \in \{1, 2, \ldots, n\}
\end{equation}
where $x = [x_1, x_2, \ldots, x_n]^T$. In this paper we are concerned with a distributed sensing scenario where each component, $x_i$, of the state vector $x$ is sensed at a different location. Although the $i^{\text{th}}$ sensor senses $x_i$ continuously in time, it transmits this data to a central controller only intermittently. In other words, the controller is a sampled-data controller that uses intermittently transmitted/sampled sensor data. In particular, we are interested in designing an asynchronous decentralized sensing mechanism based on local event-triggering that renders the origin of the closed loop system asymptotically stable.

To precisely describe the sampled-data nature of the problem, we now introduce the following notation. Let $\{t_j^{x_i}\}$ be the increasing sequence of time instants at which $x_i$ is sampled and transmitted to the controller. The resulting piecewise constant sampled signal is denoted by $x_{i,s}$, that is,
\begin{equation}
x_{i,s}(t) \triangleq x_i(t_j^{x_i}), \quad \forall t \in [t_j^{x_i}, t_{j+1}^{x_i}), \quad \forall j \in \{0, 1, 2, \ldots \} \label{eqn:samp}
\end{equation}
The sampled data, $x_{i,s}$, may be viewed as inducing an error in the the ``measurement" of the continuous-time signal, $x_i$. This measurement error is denoted by
\begin{equation*}
x_{i,e} \triangleq x_{i,s} - x_i = x_i(t_j^{x_i}) - x_i(t), \quad \forall t \in [t_j^{x_i}, t_{j+1}^{x_i})
\end{equation*}
Finally, we define the sampled-data vector and the measurement error vector as
\begin{equation*}
x_s \triangleq [x_{1,s}, x_{2,s}, \ldots, x_{n,s} ]^T, \quad \quad x_e \triangleq [x_{1,e}, x_{2,e}, \ldots, x_{n,e} ]^T
\end{equation*}
Note that, in general, the components of the vector $x_s$ are asynchronously sampled components of the plant state $x$. The components of $x_e$ are also defined accordingly.

In time-triggered implementations, the time instants $t_j^{x_i}$ are pre-determined and are commonly a multiple of a fixed sampling period. However, in event-triggered implementations the time instants $t_j^{x_i}$ are determined implicitly by a state/data based triggering condition at run-time. Consequently, an event-triggering condition may result in the inter-sample times $t_{j+1}^{x_i} - t_j^{x_i}$ to be arbitrarily close to zero or it may even result in the limit of the sequence $\{t_j^{x_i}\}$ to be a finite number (\textit{Zeno} behavior). Thus for practical utility, an event-trigger has to ensure that these scenarios do not occur.

Thus, the problem under consideration may be stated more precisely as follows. For the $n$ sensors, design event-triggers that depend only on local information and implicitly define the non-identical sequences $\{t_j^{x_i}\}$ such that (i) the origin of the closed loop system is rendered asymptotically stable and (ii) inter-sample (inter-transmission) times $t_{j+1}^{x_i} - t_j^{x_i}$ are lower bounded by a positive constant.

Finally, regarding the notation, in this paper $\norm{.}$ denotes the Euclidean norm of a vector. The state variables and other functions of time are frequently referred to without the time argument when there is no ambiguity. Similarly, some functions of variables other than time are also sometimes referred to without their arguments when there is no ambiguity.

In the next section, the main assumptions are introduced and the event-triggering for the decentralized sensing architecture is developed.

\section{Decentralized Asynchronous\\ Event-Triggering}
\label{sec:dec_async_event_trig}

In this section, the main assumptions are introduced and the event-triggers for the decentralized asynchronous sensing problem are developed.

\newcounter{saveenum}
\begin{enumerate}[label={\textbf{(A\arabic*)}},ref={A\arabic*}]
\item\label{A:ISS} Suppose $f(0,k(0)) = 0$ and that the closed loop system \eqref{eqn:nonlin_sys}-\eqref{eqn:control} is Input-to-State Stable (ISS) with respect to measurement error $x_e$. That is, there exists a smooth function $V: \mathbb{R}^n \rightarrow \mathbb{R}$ as well as class $\mathcal{K}_\infty$ functions\footnote{A continuous function $\alpha : [0, \infty) \rightarrow [0, \infty)$ is said to belong to the class $\mathcal{K}_{\infty}$ if it is strictly increasing, $\alpha(0) = 0$ and $\alpha(r) \rightarrow \infty$ as $r \rightarrow \infty$ \cite{khalil2002_book}.} $\alpha_1$, $\alpha_2$, $\alpha$ and $\gamma_i$ for each $i \in \{1, \ldots, n\}$, such that
\begin{align*}
&\alpha_1(\norm{x}) \leq V(x) \leq \alpha_2(\norm{x})\\
&\frac{\partial V}{\partial x} f(x, k(x+x_e)) \leq -\alpha(\norm{x}), \text{ if } \gamma_i(\norm{x_{i,e}}) \leq \norm{x}, \,\, \forall i.
\end{align*}
\item\label{A:lipschitz} The functions $f$, $k$ and $\gamma_i$, for each $i \in \{1, \ldots, n\}$, are Lipschitz on compact sets.
\setcounter{saveenum}{\value{enumi}}
\end{enumerate}

Note that the standard ISS assumption involves a single condition $\gamma(\norm{x_e}) \leq \norm{x}$ instead of the $n$ conditions: $\gamma_i(\norm{x_{i,e}}) \leq \norm{x}$, for $i \in \{1, \ldots, n\}$, in \eqref{A:ISS}. Given a function $\gamma(.)$ in the standard ISS assumption, one may define $\gamma_i(.)$ as
\begin{equation*}
\gamma_i(\norm{x_{i,e}}) = \gamma \left( \frac{\norm{x_{i,e}}}{\theta_i} \right), \quad i \in \{1, \ldots, n\}
\end{equation*}
where $\theta_i \in (0,1)$ such that $\displaystyle \theta^2 = \sum_{i=1}^n \theta_i^2 \leq 1$. Then, the $n$ conditions in \eqref{A:ISS} are equivalent to $\norm{x_{i,e}} \leq \theta_i \gamma^{-1}(\norm{x})$. Thus,
\begin{equation*}
\norm{x_e} = \sqrt{\sum_{i=1}^n \norm{x_{i,e}}^2} \leq \sqrt{\sum_{i=1}^n \theta_i^2 } \gamma^{-1}(\norm{x}) \leq \gamma^{-1}(\norm{x})
\end{equation*}
which is the condition in the standard ISS assumption. Similarly, given \eqref{A:ISS} one may pick $\gamma(.) = \gamma_i(.)$ for any $i$ to get the standard ISS assumption, although in practice it may be possible to choose a less conservative $\gamma(.)$.

In this section, our aim is to constructively show that decentralized asynchronous event-triggering can be used to asymptotically stabilize $x = 0$ (the trivial solution or the origin) with a desired region of attraction while also guaranteeing positive minimum inter-sample times. Further, without loss of generality, the desired region of attraction may be assumed to be a compact sub-level set $S(c)$ of the Lyapunov function $V$ in \eqref{A:ISS}. Specifically, $S(c)$ is defined as
\begin{equation}
S(c) = \{x \in \mathbb{R}^n : V(x) \leq c\} \label{eqn:S_def}
\end{equation}

\subsection{Centralized Asynchronous Event-Triggering}

The proposed design of decentralized asynchronous event-triggering progresses in stages. In the first stage, centralized event-triggers for asynchronous sampling of the sensors are proposed in the following lemma. One of the key steps in the result is choosing linear bounds on the functions $\gamma_i(.)$ on appropriately defined sets $E_i$. Given that $x \in S(c)$, we define the sets $E_i$ over which the error bounds in \eqref{A:ISS} are still satisfied, that is,
\begin{align}
E_i(c) &= \{ x_{i,e} \in \mathbb{R} : \norm{x_{i,e}} \leq \gamma_i^{-1}(\norm{x}), \,\, x \in S(c) \} \notag\\
&= \{ x_{i,e} \in \mathbb{R} : \norm{x_{i,e}} \leq \max_{x \in S(c)} \{ \gamma_i^{-1}(\norm{x}) \} \} \label{eqn:Ei_def}
\end{align}
In particular, since $x_i$ are each scalars,
\begin{equation*}
E_i(c) = [ -r_i(c), r_i(c) ], \ \ \text{with } r_i(c) = \max_{x \in S(c)} \{ \gamma_i^{-1}(\norm{x}) \}
\end{equation*}
Then, by \eqref{A:lipschitz}, for each $c \geq 0$ and each $i \in \{1, \ldots, n\}$, there exist positive constants $M_i(c)$ such that
\begin{equation}
\gamma_i( \norm{x_{i,e}} ) \leq \frac{1}{M_i(c)} \norm{x_{i,e}}, \,\, \forall x_{i,e} \in E_i(c) \label{eqn:M_def}
\end{equation}

\begin{lemma}
\label{lem:cent_async}
Consider the closed loop system \eqref{eqn:nonlin_sys}-\eqref{eqn:control} and assume \eqref{A:ISS} and \eqref{A:lipschitz} hold. Suppose that the event-triggers that determine the sampling instants, $\{ t_j^{x_i} \}$, for each $i \in \{1, \ldots, n \}$, ensure $\norm{x_{i,e}} \leq M_i(c) \norm{x}$ for all time $t \geq 0$, where $M_i(c)$ are given by \eqref{eqn:M_def} and $c \geq 0$ is an arbitrary constant. Then, the origin is asymptotically stable with $S(c)$, given by \eqref{eqn:S_def}, included in the region of attraction.
\end{lemma}

\begin{proof}
Suppose $x(0) \in S(c)$ is an arbitrary point, we have to show that the trajectory $x(.)$ asymptotically converges to zero. Next, by assumption, the sampling instants are such that for each $i \in \{1, \ldots, n \}$, $\norm{x_{i,e}} \leq M_i(c) \norm{x}$ for all time $t \geq 0$. Then, for all time $t \geq 0$, \eqref{eqn:M_def} implies
\begin{equation*}
\gamma_i( \norm{x_{i,e}} ) \leq \frac{1}{M_i(c)} \norm{x_{i,e}} \leq \norm{x}, \quad \forall x \in S(c)
\end{equation*}

Consider the ISS Lyapunov function $V(.)$ in \eqref{A:ISS}, which is a function of the state $x$. Letting
\begin{equation*}
\mathcal{M}(c) = \{ x \in S(c), \ x_e \in \mathbb{R}^n : \norm{x_{i,e}} \leq M_i(c) \norm{x}, \ \forall i \}
\end{equation*}
the time derivative of the function $V$ along the trajectories of the closed loop system, with a restricted domain, $\dot{V}(x,x_e): \mathcal{M}(c) \rightarrow \mathbb{R}$ can be upper-bounded as
\begin{equation*}
\dot{V}(x,x_e) \leq - \alpha(\norm{x}), \quad \forall [ x^T, x_e^T ]^T \in \mathcal{M}(c)
\end{equation*}
Thus, in particular, the closed loop system is dissipative on the sub-level set, $S(c)$, of the Lyapunov function $V$. Therefore, the origin is asymptotically stable with $S(c)$ included in the region of attraction.
\end{proof}

The lemma does not mention a specific choice of event-triggers but rather a family of them - all those that ensure the conditions $\norm{x_{i,e}} \leq M_i(c) \norm{x}$ are satisfied. Thus, any decentralized event-triggers in this family automatically guarantee asymptotic stability with the desired region of attraction. To enforce the conditions $\norm{x_{i,e}} \leq M_i(c) \norm{x}$ strictly, event-triggers at each sensor would need to know $\norm{x}$, which is possible only if we have centralized information. One obvious way to decentralize these conditions is to enforce $\norm{x_{i,e}} \leq M_i(c) \norm{x_i}$. However, such event-triggers cannot guarantee any positive lower bound for the inter-transmission times, which is not acceptable. So, we take an alternative approach, in which the next step is to derive lower bounds for the inter-transmission times when the conditions in Lemma \ref{lem:cent_async} are enforced strictly.

Before analyzing the lower bounds for the inter-transmission times that emerge from the event-triggers in Lemma \ref{lem:cent_async}, we introduce some notation. First, recall that under Assumption \eqref{A:ISS}, $f(0,k(0)) = 0$. Noting that for each $c \geq 0$ the set $S(c)$ contains the origin, Assumption \eqref{A:lipschitz} implies that  there exist Lipschitz constants $L(c)$ and $D(c)$ such that
\begin{equation}
\big|f(x,k(x+x_e))\big| \leq L(c) \norm{x} + D(c) \norm{x_e} \label{eqn:lip}
\end{equation}
for all $x \in S(c)$ and for all $x_e$ satisfying $\norm{x_{i,e}} \leq M_i(c) \norm{x}$, for each $i$. Similarly, there exist constants $L_i(c)$ and $D_i(c)$ for $i \in \{1, 2, \ldots, n\}$ such that
\begin{equation}
\big|f_i(x,k(x+x_e))\big| \leq L_i(c) \norm{x} + D_i(c) \norm{x_e} \label{eqn:lip_i}
\end{equation}
for all $x \in S(c)$ and for all $x_e$ satisfying $\norm{x_{i,e}} \leq M_i(c) \norm{x}$, for each $i$. Next, we introduce a function $\tau$ defined as
\begin{equation}
\tau(\omega,a_0,a_1,a_2) = \{t \geq 0 : \phi(t,0) = \omega \}\label{eqn:tau}
\end{equation}
where $a_0$, $a_1$, $a_2$ are non-negative constants and $\phi(t,\phi_0)$ is the solution of
\begin{equation*}
\dot{\phi} = a_0 + a_1 \phi + a_2 \phi^2, \quad \phi(0,\phi_0) = \phi_0
\end{equation*}
Lastly, the functions $\Gamma_i$ for $i \in \{1, \ldots, n\}$ are defined as
\begin{align}
\Gamma_i(w_i, W_i, c) \triangleq \tau(w_i,a_{0,i},a_{1,i},a_{2,i}) \label{eqn:Gamma_i}
\end{align}
where the function $\tau$ is given by \eqref{eqn:tau} and
\begin{align*}
&a_{0,i} = L_i(c) + D_i(c)W_i,\\
&a_{1,i} = L(c) + D_i(c) + D(c) W_i, \quad a_{2,i} = D(c)
\end{align*}
\begin{lemma}
\label{lem:tau_async_err}
Consider the closed loop system \eqref{eqn:nonlin_sys}-\eqref{eqn:control} and assume \eqref{A:lipschitz} holds. Let $c > 0$ be any arbitrary known constant. For $i \in \{1, \ldots, n \}$, let $0 < w_i \leq M_i(c)$ be any arbitrary constants and let $\displaystyle W_i = \sqrt{ \left( \sum_{j=1}^n w_j^2 \right) - w_i^2}$. Suppose the sampling instants are such that $\norm{x_{i,e}} \leq w_i \norm{x}$ for each $i \in \{1, \ldots, n \}$ for all time $t \geq t_0$. Finally, assume that $x(t_0)$ belongs to the compact set $S(c)$. Then, for all $t \geq t_0$, the time required for $\norm{x_{i,e}}$ to evolve from $0$ to $w_i \norm{x}$ is lower bounded by
\begin{align}
T_i &= \Gamma_i(w_i, W_i, c) > 0 \label{eqn:Ti}
\end{align}
where the functions $\Gamma_i$ are given by \eqref{eqn:Gamma_i}.
\end{lemma}

\begin{proof}
By assumption, $x(t_0)$ belongs to a known compact set $S(c)$ and $\norm{x_{i,e}} \leq w_i \norm{x} \leq M_i(c) \norm{x}$ for each $i$ for all time $t \geq t_0$. Then, Lemma \ref{lem:cent_async} guarantees that $x(t) \in S(c)$ for all time $t \geq t_0$. Thus, \eqref{eqn:lip} and \eqref{eqn:lip_i} hold for all $t \geq t_0$. Now, letting $\nu_i \triangleq \norm{x_{i,e}}/\norm{x}$ and by direct calculation we see that for $i \in \{1, \ldots, n\}$
\begin{align*}
\frac{\mathrm{d} \nu_i}{\mathrm{d}t} &= \frac{(x_{i,e}^T x_{i,e})^{-1/2} x_{i,e}^T \dot{x}_{i,e}}{\norm{x}} - \frac{x^T \dot{x} \norm{x_{i,e}}}{\norm{x}^3}\\
&\leq \frac{\norm{x_{i,e}} \norm{\dot{x}_{i,e}}}{\norm{x_{i,e}} \norm{x}} + \frac{\norm{x} \norm{\dot{x}} \norm{x_{i,e}}}{\norm{x}^3}\\
&\leq \frac{ L_i(c) \norm{x} + D_i(c)\norm{x_e}}{\norm{x}} + \frac{ \big( L(c) \norm{x} + D(c) \norm{x_e} \big) \norm{x_{i,e}}}{\norm{x}^2}
\end{align*}
where for $x_{i,e} = 0$ the relation holds for all directional derivatives. Next, notice that
\begin{align*}
\frac{ \norm{x_e}}{\norm{x}} &= \sqrt{\sum_{j=1}^{j=n}\nu_j^2} \leq \sqrt{ \left( \sum_{j=1}^{j=n}w_j^2 \right) - w_i^2 + \nu_i^2} \leq W_i + \nu_i
\end{align*}
where the condition that $\nu_i \leq w_i$, the definition of $W_i$ and the triangle inequality property have been utilized. Thus,
\begin{align*}
\frac{\mathrm{d} \nu_i}{\mathrm{d}t} &\leq L_i(c) + L(c) \nu_i + \big( D_i(c) + D(c) \nu_i \big) (W_i + \nu_i)\\
&= a_{0,i} + a_{1,i} \nu_i + a_{2,i} \nu_i^2
\end{align*}
The claim of the Lemma now directly follows.
\end{proof}

Now, by combining Lemmas \ref{lem:cent_async} and \ref{lem:tau_async_err}, we get the following result for the centralized asynchronous event-triggering.
\begin{theorem}
\label{thm:central_async}
Consider the closed loop system \eqref{eqn:nonlin_sys}-\eqref{eqn:control} and assume \eqref{A:ISS}-\eqref{A:lipschitz} hold. Suppose the $i^{\text{th}}$ sensor transmits its measurement to the controller whenever $\norm{x_{i,e}} \geq w_i \norm{x}$, where $0 < w_i \leq M_i(c)$, with $M_i(c)$ given by \eqref{eqn:M_def} and $c \geq 0$ any arbitrary constant. Then, the origin is asymptotically stable with $S(c)$ included in the region of attraction and the inter-transmission times of each sensor have a positive lower bound given by $T_i$ in \eqref{eqn:Ti}.
\end{theorem}
\begin{proof}
The triggering conditions ensure that $\norm{x_{i,e}} \leq w_i \norm{x} \leq M_i(c) \norm{x}$ for all $t > 0$. Thus, Lemma \ref{lem:cent_async} guarantees $x(t) \in S(c)$ for all $t \geq 0$ and that the origin is asymptotically stable with $S(c)$ included in the region of attraction. Since $S(c)$ is positively invariant, Lemma \ref{lem:tau_async_err} guarantees a positive lower bound for the inter-transmission times.
\end{proof}

\subsection{Decentralized Asynchronous Event-Triggering}
\label{sec:1way_com}

Now, turning to the main subject of this paper, in the decentralized sensing case, unlike in the centralized sensing case, no single sensor has knowledge of the exact value of $\norm{x}$ from the locally sensed data. We may let the event-trigger at the $i^{\text{th}}$ sensor enforce the more conservative condition $\norm{x_{i,e}} \leq w_i \norm{x_i}$ and still satisfy the assumptions of Lemma \ref{lem:cent_async}, though such a choice cannot guarantee a positive minimum inter-sample time. We overcome this problem through the following observation. The event-triggers in Theorem \ref{thm:central_async},
\begin{align}
&t_{j+1}^{x_i} = \min \{t \geq t_j^{x_i} : \norm{x_{i,e}} \geq w_i \norm{x} \}, \,\, i \in \{1,\ldots, n\} \label{eqn:cent_async_trig1}
\end{align}
can be equivalently expressed as
\begin{align}
&t_{j+1}^{x_i} = \min \{t \geq t_j^{x_i} + T_i : \norm{x_{i,e}} \geq w_i \norm{x} \} \label{eqn:cent_async_trig2}
\end{align}
where $T_i$ are the estimates of positive inter-sample times provided by Lemma \ref{lem:tau_async_err} in \eqref{eqn:Ti}. In the latter interpretation, a minimum dwell time is explicitly enforced, only after which, the state based condition is checked. Now, in order to let the event-triggers depend only on locally sensed data, one can let the sampling times, for $i \in \{1, \ldots, n\}$, be determined as
\begin{align}
&t_{j+1}^{x_i} = \min \{t \geq t_j^{x_i} + T_i : \norm{x_{i,e}} \geq w_i \norm{x_i} \} \label{eqn:event_trig_dec}
\end{align}
where $T_i$ are given by \eqref{eqn:Ti}. This allows us to implement decentralized asynchronous event-triggering.

\begin{remark}
\label{rem:wang_compare}
The event-triggers proposed in \cite{wang2009dist, wang2011} take the form of \eqref{eqn:event_trig_dec} with $T_i = 0$. Of course the parameters $w_i$ are computed in a different manner. Nevertheless, event-triggers of the form \eqref{eqn:event_trig_dec}, with $T_i = 0$, do not in general guarantee a positive lower bound for the inter-transmission times in the global or even the semi-global sense irrespective of whether the system is linear or nonlinear. This is due to the fact that, even under a weak coupling assumption, the set $\{x \in \mathbb{R}^n: x_i = 0\}$ is in general not an equilibrium set.

In \cite{wang2011} a modified version of the event-triggers is also presented, which may be interpreted as including timers. However, in contrast to our proposed scheme, the timers in \cite{wang2011} enforce a maximum inter-transmission time.
\end{remark}

\begin{remark}
\label{rem:compare_mazo2011}
Although at first sight our approach of explicitly enforcing a lower bound on inter-transmission times may seem similar to that of \cite{mazo2011}, there are important differences. The most important difference stems from the fact that in \cite{mazo2011}, the control input to the plant is always based on synchronously sampled data of the decentralized sensors. This allows the controller to utilize the lower threshold for inter-event times of the centralized control system, as that of \cite{tabuada2007}, to reject very closely timed requests by the sensors.

In contrast, our proposed methodology allows the controller to rely only on asynchronously sampled data. Further, the event-trigger at each of the sensors utilizes a lower threshold for the inter-transmission times. These differences necessitate the computation of the lower thresholds for the inter-transmission times as in Lemma \ref{lem:tau_async_err} instead of as in \cite{tabuada2007}.
\end{remark}

The following theorem is the core result of this paper and it shows that by appropriately choosing the constants $T_i$ and $w_i$, the event triggers, \eqref{eqn:event_trig_dec}, guarantee asymptotic stability of the origin while also explicitly enforcing a positive lower bound for inter-sample times.

\begin{theorem}
\label{thm:dec}
Consider the closed loop system \eqref{eqn:nonlin_sys}-\eqref{eqn:control} and assume \eqref{A:ISS} and \eqref{A:lipschitz} hold. Let $c \geq 0$ be an arbitrary known constant. For each $i \in \{1, 2, \ldots, n\}$, let $w_i$ be a positive constant such that $w_i \leq M_i(c)$, where $M_i(c)$ is given by \eqref{eqn:M_def} and $T_i$ be given by \eqref{eqn:Ti}. Suppose the sensors asynchronously transmit the measured data at time instants determined by \eqref{eqn:event_trig_dec} and that $t_0^{x_i} \leq 0$ for each $i \in \{1, 2, \ldots, n\}$. Then, the origin is asymptotically stable with $S(c)$ included in the region of attraction and the inter-transmission times of each sensor are explicitly enforced to have a positive lower threshold.
\end{theorem}
\begin{proof}
The statement about the positive lower threshold for inter-transmission times is obvious from \eqref{eqn:event_trig_dec} and only asymptotic stability remains to be proven. This can be done by showing that the event-triggers \eqref{eqn:event_trig_dec} are included in the family of event-triggers considered in Lemma \ref{lem:cent_async}. From the equivalence of \eqref{eqn:cent_async_trig1} and \eqref{eqn:cent_async_trig2}, it is clearly true that $\norm{x_{i,e}} \leq w_i \norm{x}$ for $t \in [t_j^{x_i}, t_j^{x_i} + T_i]$, for each $i \in \{1, 2, \ldots, n\}$ and each $j$. Next, for $t \in [t_j^{x_i} + T_i, t_{j+1}^{x_i}]$, \eqref{eqn:event_trig_dec} enforces $\norm{x_{i,e}} \leq w_i \norm{x_i}$, which implies $\norm{x_{i,e}} \leq w_i \norm{x}$ since $\norm{x_i} \leq \norm{x}$. Therefore, the event-triggers in \eqref{eqn:event_trig_dec} are included in the family of event-triggers considered in Lemma \ref{lem:cent_async}. Hence, $x \equiv 0$ (the origin) is asymptotically stable with $S(c)$ included in the region of attraction.
\end{proof}

\begin{remark}
Although the assumption that $t_0^{x_i} \leq 0$, for each $i$, in Theorem \ref{thm:dec} has not been used in the proof explicitly, it serves two key purposes - avoiding having the sensors send their first transmissions of data synchronously; and for the controller to have some latest sensor data to compute the controller output at $t = 0$.
\end{remark}

\begin{remark}
In Theorem \ref{thm:dec}, the parameters $w_i$ cannot be chosen in a decentralized manner unless $M_i(c)$ and hence $c$ is fixed \textit{a priori}. In other words, the desired region of attraction $S(c)$ has to be chosen at the time of the system installation. This can potentially lead to the parameters $w_i$ to be chosen conservatively to guarantee a larger region of attraction. One possible solution is to let the central controller communicate the parameters $w_i$ to the sensors at $t = 0$. In any case, for $t > 0$, the sensors only have to transmit their data to the controller and not have to receive any communication.
\end{remark}

Apart from the fact that the set $S(c)$ is chosen \textit{a priori}, conservativeness in transmission frequency may also be introduced because the Lipschitz constants of the  nonlinear functions $\gamma_i(.)$, \eqref{eqn:M_def}, are not updated after their initialization despite knowing that the system state is progressively restricted to smaller and smaller subsets of $S(c)$. Although we started from the idea that energy may be saved by making sure that sensors do not have to listen, the cost of increased transmissions may not be in its favor. Thus, we now describe a design where the central controller intermittently communicates updated $w_i$ and $T_i$ to the event-triggers.

\subsection{Decentralized Asynchronous Event-Triggering with Intermittent Communication from the Central Controller}
\label{sec:2way_com}

The basic idea of the design with bi-directional communication (sensors to controller and controller to the sensors) is simple. Whenever the central controller determines that the state is confined to a sufficiently small positively invariant subset of $S(c)$, the parameters $w_i$, $T_i$ and $c$ are updated, which the sensor nodes use as in the previous subsection. Of course, the controller has access only to the asynchronously sampled data, $x_s$. As a result, the controller can only determine an upper bound for the Lyapunov function of the system state.

Recall from the proof of Theorem \ref{thm:dec} that the event-triggers \eqref{eqn:event_trig_dec} ensure that $x(t) \in \mathcal{R}(x_s(t))$ for all $t \geq 0$, where
\begin{equation}
\mathcal{R}(x_s) \triangleq \{ x \in \mathbb{R}^n : \norm{x_{i,s} - x_i} \leq w_i \norm{x}, \ \forall i \in \{1, \ldots, n\} \} \label{eqn:calR}
\end{equation}
Since the controller knows the parameters used by each event-trigger, it may compute an upper bound for the Lyapunov function of the system state, given $x_s$, as
\begin{equation}
\mathcal{V}(x_s) \geq \max \{ V(x) : x \in \mathcal{R}(x_s) \} \label{eqn:calV}
\end{equation}
Note that the maximum exists and is finite for finite $x_s$ because $V(x)$ is a continuous function of $x$ and $\mathcal{R}$ is a closed and bounded set.

Now, let $\{ t_j^{\mathcal{V}} \}$ be the sequence of time instants at which $\mathcal{V}$ is sampled and the sensor event-trigger parameters $w_i$ and $T_i$ are updated. Then the idea here is to determine the sequence $\{ t_j^{\mathcal{V}} \}$ by an event-trigger running at the central controller, namely,
\begin{equation}
t_{j+1}^{\mathcal{V}} = \min \{t \geq t_j^{\mathcal{V}} + \mathcal{T} : \mathcal{V}(x_s(t)) \leq \rho \mathcal{V}(x_s(t_j^\mathcal{V})) \} \label{eqn:ts_calV}
\end{equation}
where $\mathcal{T} > 0$ (a positive dwell time) and $0 < \rho < 1$ are arbitrary constants. The initial condition $t_0^{\mathcal{V}} = 0$ and $\mathcal{V}(x_s(t_0^{\mathcal{V}})) = c$ may be chosen, where $c$ determines the region of attraction $S(c)$. Thus, with a slight abuse of notation, the `sampled' version of $\mathcal{V}$ is denoted by
\begin{equation}
\mathcal{V}_s(t) \triangleq \mathcal{V}(x_s(t_j^{\mathcal{V}})), \ \forall t \in [t_j^{\mathcal{V}}, t_{j+1}^{\mathcal{V}}), \ \mathcal{V}_s(0) = c \label{eqn:calVs}
\end{equation}
where $c > 0$ is an arbitrary constant, $t_j^{\mathcal{V}}$ are given by \eqref{eqn:ts_calV} and $\mathcal{V}$ is given by \eqref{eqn:calV}. Note that since $w_i$ and $T_i$ are updated at the time instants $t_j^{\mathcal{V}}$, they are time varying, piece-wise constant parameters, that is
\begin{equation}
w_i(t) = w_i(t_j^{\mathcal{V}}), \  T_i(t) = T_i(t_j^{\mathcal{V}}), \ \forall t \in [t_j^{\mathcal{V}}, t_{j+1}^{\mathcal{V}}) \label{eqn:wT_timevar}
\end{equation}
Note that the specification of $w_i(t_j^{\mathcal{V}})$ and $T_i(t_j^{\mathcal{V}})$ still remains. Let us also define the time varying piecewise constant auxiliary signals
\begin{equation}
W_i(t) = \sqrt{ \left( \sum_{j=1}^n w_j^2(t) \right) - w_i^2(t)} \label{eqn:Wi}
\end{equation}
The sensor transmissions are assumed to be triggered by the time varying analogue of \eqref{eqn:event_trig_dec}, that is
\begin{align}
&t_{j+1}^{x_i} = \min \{t \geq t_j^{x_i} + T_i(t) : \norm{x_{i,e}(t)} \geq w_i(t) \norm{x_i(t)} \} \label{eqn:ET_dec_sensors_bidir}
\end{align}

As one might expect from the results of the previous subsection, between any two updates of $\mathcal{V}_s$, that is for $t \in [t_j^{\mathcal{V}}, t_{j+1}^{\mathcal{V}})$ for any $j$, the Lyapunov function is guaranteed to decrease monotonously and the inter-transmission times of the $i^{\text{th}}$ sensor are lower bounded by $T_i(t_j^{\mathcal{V}}) = \Gamma_i(w_i(t_j^{\mathcal{V}}), W_i(t_j^{\mathcal{V}}), \mathcal{V}(t_j^{\mathcal{V}}))$, where the function $\Gamma_i$ is given by \eqref{eqn:Gamma_i}.

\begin{remark}
\label{rem:Gamma_i}
Note from the definition of $\Gamma_i$ that $\Gamma_i$ increases as $w_i$ increases, and as $W_i$ or $\mathcal{V}$ decrease.
\end{remark}

Thus, although $T_i(t_j^{\mathcal{V}}) > 0$ for each $i$ and each $j$, in the absence of further information one cannot rule out the possibility of $\displaystyle \lim_{j \rightarrow \infty} T_i(t_j^{\mathcal{V}}) \rightarrow 0$. Hence, let us make the following assumption.

\begin{enumerate}[label={\textbf{(A\arabic*)}},ref={A\arabic*}]
\setcounter{enumi}{\value{saveenum}}
\item\label{A:lim_gamma_i} $\displaystyle \lim_{s \rightarrow 0} ( \gamma_i(s)/s ) \rightarrow ( 1/g_i ) > 0$ for each $i \in \{1, \ldots, n\}$.
\end{enumerate}

\begin{remark}
\label{rem:M_L_D}
As $S(c_1) \subset S(c_2)$ if $c_1 \leq c_2$, $M_i(.)$ in \eqref{eqn:M_def} are assumed, without loss of generality, to be non-increasing functions of $c$. Similarly, $L_i(c)$, $D_i(c)$, $L(c)$ and $D(c)$ (defined in \eqref{eqn:lip}, \eqref{eqn:lip_i}) are assumed, without loss of generality, to be non-decreasing functions of $c$.
\end{remark}

\begin{theorem}
\label{thm:dec_wT}
Consider the closed loop system \eqref{eqn:nonlin_sys}-\eqref{eqn:control}, \eqref{eqn:ts_calV} and \eqref{eqn:ET_dec_sensors_bidir}. Suppose \eqref{A:ISS} and \eqref{A:lipschitz} hold. For each $i \in \{1, 2, \ldots, n\}$, let $w_i(t)$ and $T_i(t)$ be positive piecewise-constant signals satisfying \eqref{eqn:wT_timevar} with $T_i(t) = \Gamma_i(w_i(t), W_i(t), \mathcal{V}_s(t))$. Further, for each $i$, let $w_i(t)$ be non-decreasing in time satisfying $0 < w_i(t) \leq M_i(\mathcal{V}_s(t))$ for all $t \geq 0$. Suppose the sensors asynchronously transmit the measured data at time instants determined by \eqref{eqn:ET_dec_sensors_bidir} and that $t_0^{x_i} \leq 0$ for each $i \in \{1, 2, \ldots, n\}$. Then, the origin is asymptotically stable with $S(c)$ included in the region of attraction.
\end{theorem}

\begin{proof}
First, notice that $\mathcal{V}_s(t)$ is non-increasing in time. This implies that for each $i$,  $M_i(\mathcal{V}_s(t))$ is non-decreasing in time and a non-decreasing $w_i(t)$ satisfying $w_i(t) \leq M_i(\mathcal{V}_s(t))$ exists. Clearly, the Lyapunov function evaluated at the state of the system is at all times lesser than the piecewise constant and non-increasing signal $\mathcal{V}_s$. Thus, $x \in S(\mathcal{V}_s)$ at all times, where $S(.)$ is given by \eqref{eqn:S_def}. Given that $T_i(t) = \Gamma_i(w_i(t), W_i(t), \mathcal{V}_s(t))$ and sensor event-triggers given by \eqref{eqn:ET_dec_sensors_bidir}, one can conclude from arguments similar to those in Theorem \ref{thm:dec} that the origin of the closed loop system is asymptotically stable with $S(\mathcal{V}_s(0)) = S(c)$ included in the region of attraction.
\end{proof}

From \eqref{eqn:ts_calV}, it is clear that the inter-transmission times $\{ t_{j+1}^{\mathcal{V}} - t_j^{\mathcal{V}} \}$ are lower bounded by $\mathcal{T} > 0$. We would also like the inter-transmission times of each sensor to have a uniform positive lower bound.

\begin{corollary}[\textbf{Corollary to Theorem \ref{thm:dec_wT}}]
In addition, assume \eqref{A:lim_gamma_i} holds. Then, the inter-transmission times of each sensor has a uniform positive lower bound.
\end{corollary}

\begin{proof}
Since by assumption $w_i(t) > 0$ for all $t \geq 0$, Lemma \ref{lem:tau_async_err} guarantees that $T_i(t) = \Gamma_i(w_i(t), W_i(t), \mathcal{V}_s(t)) > 0$. Now we show that $T_i(t)$ has a uniform positive lower bound. Since the signal $\mathcal{V}_s$ is non-increasing, Remark \ref{rem:M_L_D} guarantees that $M_i(\mathcal{V}_s(t))$ are non-decreasing in time. Again by assumption in Theorem \ref{thm:dec_wT}, $w_i(t) > 0$ is non-decreasing in time such that $w_i(t) \leq M_i(\mathcal{V}_s(t))$ (note that such a $w_i(t)$ always exists). In other words, $w_i(t) \geq w_i(0)$ for all $t \geq 0$. Next, Assumption \eqref{A:lim_gamma_i} means that $\displaystyle \lim_{c \rightarrow 0} M_i(c) \leq g_i$. Then, from the definition of $M_i(c)$, in \eqref{eqn:M_def}, $M_i(c) \leq g_i$ for all $c \geq 0$. Thus, $w_i(t) \leq g_i$ for all $t \geq 0$ and there exist constants $h_i > 0$ such that $W_i(t) \leq h_i$ for all $t \geq 0$. Finally, Remark \ref{rem:Gamma_i} implies that for all $t \geq 0$,
\begin{equation*}
T_i(t) = \Gamma_i(w_i(t), W_i(t), \mathcal{V}_s(t)) \geq \Gamma_i(w_i(0), h_i, \mathcal{V}_s(0)) > 0
\end{equation*}
\end{proof}

\begin{remark}
Assumption \eqref{A:lim_gamma_i} is not restrictive and can be easily overcome if it is not satisfied, under the assumption that \eqref{A:lipschitz} holds. Assumption \eqref{A:lipschitz} implies that for each $\bar{s} > 0$, there exists a $g_i(\bar{s}) > 0$ such that $\gamma_i(s) \leq ( 1/g_i(\bar{s}) ) s$, for all $s \in [0, \bar{s}]$. If a function $\gamma_i$ does not satisfy the Assumption \eqref{A:lim_gamma_i} then one can simply let $\gamma_i(s) = ( 1/g_i(\bar{s}) ) s$, for all $s \in [0, \bar{s}]$ for some arbitrary $\bar{s} > 0$.
\end{remark}

\begin{remark}
Computing the upper bound on $V$, \eqref{eqn:calV}, may be computationally intensive depending on the Lyapunov function and the dimension of the system. However, since the Lyapunov function is guaranteed to decrease even with no updates to $w_i$ and $T_i$, there is no restriction on the time needed to compute the upper bound on $V$ and to update the parameters of the event-triggers. Offline computation and mapping different regions of state space with different values of $\mathcal{V}$ is also possible. On the other hand, it is true that the updates to all the event-triggers have to occur synchronously.
\end{remark}

\begin{remark}
\label{rem:wT}
Note that the aim of the event-triggers \eqref{eqn:ET_dec_sensors_bidir} is to ``approximately" enforce the conditions $\norm{x_{i,e}} \leq w_i(t) \norm{x}$. Thus whenever $w_i$ and $T_i$ are updated, the new parameters in the event-triggers are consistent with and are likely to be an improvement over the previous parameters. Although $w_i$ can be chosen to be non-decreasing in time, the same cannot be said about $T_i$, at least not without some restriction on $w_j(t)$ for each $j$.
\end{remark}

From \eqref{eqn:M_def}, \eqref{eqn:lip} and \eqref{eqn:lip_i}, it is clear that $M_i(c)$, $D(c)$ and $D_i(c)$ each capture very similar properties of the system - tolerance or sensitivity to the measurement errors in the whole set $S(c)$. Since $S(c_1) \subset S(c_2)$ if $c_1 \leq c_2$, the ``tolerance to measurement errors", $M_i(c)$, increases as $c$ gets smaller. Alternatively, ``sensitivity to the measurement errors", $D(c)$ and $D_i(c)$, decreases as $c$ gets smaller. As a result, non-decreasing $w_i(t)$ intuitively suggests non-decreasing $T_i(t)$ and non-decreasing ``average inter-transmission times". However, this statement is far from apparent in the logical sense, mainly because of the asynchronous transmissions. Nonetheless we can still say something concrete in certain cases.

\begin{corollary}[\textbf{Corollary to Theorem \ref{thm:dec_wT}}]
\label{cor:wT}
For as long as $D_i(\mathcal{V}_s(t)) W_i(t)$ and $D(\mathcal{V}_s(t)) W_i(t)$ are each non-increasing, $T_i(t)$ is non-decreasing.
\end{corollary}

\begin{proof}
As $\mathcal{V}_s(t)$ is a non-increasing signal, $L(\mathcal{V}_s(t))$, $D(\mathcal{V}_s(t))$, $L_i(\mathcal{V}_s(t))$ and $D_i(\mathcal{V}_s(t))$ are in turn non-increasing in time. In the definition of $\Gamma_i$, \eqref{eqn:Gamma_i}, it is seen that $W_i(t)$ appears in the terms $D_i(\mathcal{V}_s(t)) W_i(t)$ and $D(\mathcal{V}_s(t)) W_i(t)$. In addition $w_i(t) \leq M_i(\mathcal{V}_s(t))$ has been assumed to be non-decreasing. Thus, a simple application of the Comparison Lemma \cite{khalil2002_book} implies that $T_i(t)$ as defined in Theorem \ref{thm:dec_wT} is non-decreasing.
\end{proof}

Given a nonlinear system, a systematic characterization of the region of state space where the conditions of Corollary \ref{cor:wT} hold is an interesting problem and will be pursued in future. On the other hand, one may be able to choose $w_i(t) \leq M_i(\mathcal{V}_s(t))$ for each $i$ such that $T_i(t)$ are non-decreasing in time. The following result demonstrates the existence of $w_i(t)$ such that $T_i(t)$ is non-decreasing for each $i$, while also guaranteeing asymptotic stability.

\begin{corollary}[\textbf{Corollary to Theorem \ref{thm:dec_wT}}]
For each $i$, there exists a sequence of updates of $ w_i(t_j^{\mathcal{V}}) $ such that $T_i(t) = \Gamma_i(w_i(t), W_i(t), \mathcal{V}_s(t))$ is non-decreasing in time.
\end{corollary}

\begin{proof}
Since for each $j$, $\mathcal{V}(t_{j+1}^{\mathcal{V}}) < \mathcal{V}(t_j^{\mathcal{V}})$, a consequence of Remark \ref{rem:Gamma_i} is that
\begin{equation*}
\Gamma_i(w_i(t_j^{\mathcal{V}}), W_i(t_j^{\mathcal{V}}), \mathcal{V}(t_{j+1}^{\mathcal{V}})) \geq \Gamma_i(w_i(t_j^{\mathcal{V}}), W_i(t_j^{\mathcal{V}}), \mathcal{V}(t_j^{\mathcal{V}}))
\end{equation*}
Thus for each $i$, letting $w_i(t) \equiv w_i(0) \leq M_i(c)$ and consequently, $W_i(t) \equiv W_i(0)$ the corollary is seen to be true.
\end{proof}

In general for nonlinear systems, $\Gamma_i(w_i, W_i, \mathcal{V}(t_{j+1}^{\mathcal{V}}))$ would be strictly increasing with decreasing $\mathcal{V}(t_{j+1}^{\mathcal{V}})$. Further, although the corollary has been proved using the trivial choice of $w_i(t) \equiv w_i(0)$ and $W_i(t) \equiv W_i(0)$, in practice one could choose $\{ w_i(t_{j+1}^{\mathcal{V}}) \}_{i=1}^n $ as
\begin{equation*}
w_i(t_{j+1}^{\mathcal{V}}) \in [ w_i(t_j^{\mathcal{V}}), M_i(t_j^{\mathcal{V}}) ], \ \text{s.t. } T_i(t_{j+1}^{\mathcal{V}}) \geq T_i(t_j^{\mathcal{V}}), \ \forall i
\end{equation*}
Further, for each $i$, $T_i(t_{j+1}^{\mathcal{V}})$ could be chosen to be as large as possible. Clearly, this is a multi-objective resource allocation problem and has to be studied rigorously. Further, it may also be possible to design a single-objective function based on the aim of overall reduction of inter-transmission times.

Note that we have not really addressed mathematically the question of whether there is an improvement in the average inter-transmission times. Non-decreasing $T_i(t)$ do not necessarily on their own mean longer average inter-transmission times. In fact, comparing across different choices of the parameters of the control system is equivalent to comparing different closed loop systems. Although the asynchronous nature of the transmissions is a major impediment to the quantification of any improvement in the inter-transmission times, at a more fundamental level it is due to the lack of tools to compare the ``average" inter-transmission times across different aperiodic sampled-data control systems. The development of such tools is a problem in its own right and not within the scope of this paper.

\section{Linear Time Invariant Systems}
\label{sec:lin_sys}

Now, let us consider the special case of Linear Time Invariant (LTI) systems with quadratic Lyapunov functions. Thus, the system dynamics may be written as
\begin{align}
\dot{x} &= A x + B u, \quad x \in \mathbb{R}^n, \quad u \in \mathbb{R}^m \label{eqn:lin_sys}\\
u &= K(x + x_e) \label{eqn:lin_con}
\end{align}
where $A$, $B$ and $K$ are matrices of appropriate dimensions. As in the general case, let us assume that for each $i \in \{1, 2, \ldots, n\}$, $x_i \in \mathbb{R}$ is sensed by the $i^{\text{th}}$ sensor. Comparing with \eqref{eqn:lin_sys}-\eqref{eqn:lin_con} we see that $x_i$ evolves as
\begin{equation}
\dot{x}_i = r_i(A) x + r_i(BK) (x + x_e)
\end{equation}
where the notation $r_i(H)$ denotes the $i^{\text{th}}$ row of the matrix $H$. Also note that $x_e$ and $x_{i,e}$ are defined just as in Section \ref{sec:prob_setup}.

Now, suppose the matrix $(A+BK)$ is Hurwitz, which is equivalent to the following statement.
\begin{enumerate}[label={\textbf{(A\arabic*)}},ref={A\arabic*}]
\setcounter{enumi}{\value{saveenum}}
\item\label{A:lyap_eq} Suppose that for any given symmetric positive definite matrix $Q$, there exists a symmetric positive definite matrix $P$ such that $P(A+BK) + (A+BK)^T P = - Q$.
\end{enumerate}

Then, the following result is obtained as a special case of Theorem \ref{thm:dec} and prescribes the constants $T_i$ and $w_i$ in the event triggers, \eqref{eqn:event_trig_dec}, that guarantee global asymptotic stability of the origin while also explicitly enforcing a positive minimum inter-sample time.

\begin{theorem}
\label{thm:lin_dec}
Consider the closed loop system \eqref{eqn:lin_sys}-\eqref{eqn:lin_con} and assume \eqref{A:lyap_eq} holds. Let $Q$ be any symmetric positive definite matrix and let $Q_m$ be the smallest eigenvalue of $Q$. For each $i \in \{1, 2, \ldots, n\}$, let
\begin{align}
&\theta_i \in (0,1)\quad \text{ s.t. }\,\, \displaystyle \theta = \sum_{i=1}^n \theta_i \leq 1 \label{eqn:theta_i}\\
&\displaystyle w_i = \frac{\sigma \theta_i Q_m}{\norm{c_i(2PBK)}} \label{eqn:wi_lin}
\end{align}
where $0 < \sigma < 1$ is a constant and $c_i(2PBK)$ is the $i^{\text{th}}$ column of the matrix $(2PBK)$. Let $T_i$ be defined as
\begin{align}
&T_i = \tau(w_i,a_0,a_1,a_2) \label{eqn:Ti_lin}
\end{align}
where the function $\tau$ is given by \eqref{eqn:tau} and
\begin{align*}
&a_0 = \norm{r_i(A+BK)} + \norm{r_i(BK)}W_i, \\
& a_1 = \norm{A+BK} + \norm{r_i(BK)} + \norm{BK} W_i, \quad a_2 = \norm{BK}
\end{align*}
where $\displaystyle W_i = \sqrt{ \left( \sum_{j=1}^n w_j^2 \right) - w_i^2}$. Suppose the sensors asynchronously transmit the measured data at time instants determined by \eqref{eqn:event_trig_dec}. Then, the origin is globally asymptotically stable and the inter-transmission times are explicitly enforced to have a positive lower threshold. \qed
\end{theorem}

In the context of the results for nonlinear systems in Section \ref{sec:dec_async_event_trig}, the reason we are able to achieve global asymptotic stability for LTI systems is because, the system dynamics, the functions $\gamma_i(.)$ are globally Lipschitz, thus giving us constants $w_i$  and $T_i$ that hold globally. In fact, for linear systems, something more is ensured - the proposed asynchronous event-triggers guarantee a type of scale invariance.

Scaling laws of inter-execution times for centralized \textit{synchronous} event-triggering have been studied in \cite{anta2010}. In particular, Theorem 4.3 of \cite{anta2010}, in the special case of linear systems, guarantees scale invariance of the inter-execution times determined by a centralized event-trigger $\norm{x_e} = W \norm{x}$. The centralized and decentralized asynchronous event-triggers developed in this paper are under-approximations of this kind of central event-triggering. In the following, we show that the scale invariance is preserved in the asynchronous event-triggers. As an aside, we would like to point out that the decentralized event-triggers proposed in \cite{mazo2011async, mazo2012, mazo2012arxiv} are not scale invariant. In order to precisely state the notion of scale invariance and to state the result the following notation is useful. Let $x(t)$ and $z(t)$ be two solutions to the system: \eqref{eqn:lin_sys}-\eqref{eqn:lin_con} along with the event-triggers \eqref{eqn:event_trig_dec}.

\begin{theorem}
\label{thm:lin_scale_inv}
Consider the closed loop system \eqref{eqn:lin_sys}-\eqref{eqn:lin_con} and assume \eqref{A:lyap_eq} holds. Let $Q$ be any symmetric positive definite matrix and let $Q_m$ be the smallest eigenvalue of $Q$. For each $i \in \{1, 2, \ldots, n\}$, let $\theta_i$, $w_i$ and $T_i$ be defined as in \eqref{eqn:theta_i}, \eqref{eqn:wi_lin} and \eqref{eqn:Ti_lin}, respectively. Suppose the sensors asynchronously transmit the measured data at time instants determined by \eqref{eqn:event_trig_dec}. Assuming $b$ is any scalar constant, let $[z(0)^T, z_s(0)^T]^T = b [x(0)^T, x_s(0)^T]^T \in \mathbb{R}^n \times \mathbb{R}^n$ be two initial conditions for the system. Further let $t_0^{z_i} = t_0^{x_i} \leq 0$ for each $i \in \{1, \ldots, n\}$. Then, $[z(t)^T, z_s(t)^T]^T = b [x(t)^T, x_s(t)^T]^T$ for all $t \geq 0$ and $t_j^{x_i} = t_j^{z_i}$ for each $i$ and $j$.
\end{theorem}

\begin{proof}
First of all, let us introduce two strictly increasing sequences of time, $\{t_j^{z_s}\}$ and $\{t_j^{x_s}\}$, at which one or more components of $z_s$ and $x_s$ are updated, respectively. Further, without loss of generality, assume $t_0^{z_s} = t_0^{x_s}$. The proof proceeds by mathematical induction. Let us suppose that $t_j^{z_s} = t_j^{x_s} = t_j$ for each $j \in \{0, \ldots, k\}$ and that $[z(t)^T, z_s(t)^T]^T = b [x(t)^T, x_s(t)^T]^T$ for all $t \in [0, t_k)$. Then, letting $\underline{t}_{k+1} = \min \{ t_{k+1}^{z_s}, t_{k+1}^{x_s} \}$ the solution, $z$, in the time interval $[t_k, \underline{t}_{k+1})$ satisfies
\begin{align*}
z(t) &= e^{A(t-t_k)} z(t_k) + \int_{t_k}^{t} e^{A(t-\sigma)} BK z_s(t_k) \mathrm{d} \sigma \\
&= b e^{A(t-t_k)} x(t_k) + b \int_{t_k}^{t} e^{A(t-\sigma)} BK x_s(t_k) \mathrm{d} \sigma
\end{align*}
Hence,
\begin{equation}
z(t) = b x(t), \quad \forall t \in [t_k, \underline{t}_{k+1}) \label{eqn:scal_inv_x}
\end{equation}
Further, in the time interval $[t_k, \underline{t}_{k+1})$
\begin{align}
z_{i,e}(t) = z_i(t_k) - z_i(t) = b ( x_i(t_k) - x_i(t) ) = b x_{i,e}(t) \label{eqn:scal_inv_xe}
\end{align}
Similarly, for all $t \in [t_k, \underline{t}_{k+1})$,
\begin{equation}
\frac{\norm{z_{i,e}(t)}}{\norm{z(t)}} = \frac{\norm{x_{i,e}(t)}}{\norm{x(t)}} \label{eqn:scal_inv}
\end{equation}

Without loss of generality, assume $z_{i,s}$ is updated at $\underline{t}_{k+1}$. Then, clearly, at least $T_i$ amount of time has elapsed since $z_{i,s}$ was last updated. Next, by the assumption that $t_0^{z_i} = t_0^{x_i} \leq 0$ and the induction statement, it is clear that at least $T_i$ amount of time has elapsed since $x_{i,s}$ also was last updated. Further, it also means that $\norm{z_{i,s}(t_k) - z_i(\underline{t}_{k+1})} \geq w_i \norm{z_i(\underline{t}_{k+1})}$. Then, \eqref{eqn:scal_inv_x}-\eqref{eqn:scal_inv_xe} imply that $\norm{x_{i,s}(t_k) - x_i(\underline{t}_{k+1})} \geq w_i \norm{x_i(\underline{t}_{k+1})}$, meaning $\underline{t}_{k+1} = t_{k+1}^{z_s} = t_{k+1}^{x_s} = t_{k+1}$. Arguments analogous to the preceding also hold for multiple $z_{i,s}$ updated at $\underline{t}_{k+1}$ instead of one or even $x_{i,s}$ instead of $z_{i,s}$. Since the induction statement is true for $k = 0$, we conclude that the statement of theorem is true.
\end{proof}

\begin{remark}
From the proof of Theorem \ref{thm:lin_scale_inv}, \eqref{eqn:scal_inv} specifically, it is clear that the centralized asynchronous event-triggers \eqref{eqn:cent_async_trig1} also guarantee scale invariance.
\end{remark}

\begin{remark}
Scale invariance, as described in Theorem \ref{thm:lin_scale_inv}, means that the inter-transmission times over an arbitrary length of time is independent of the scale (or the magnitude) of the initial condition of the system. Similarly for any given scalar, $0 < \delta < 1$, the time and the number of transmissions it takes for $\norm{x(t)}$ to reduce to $\delta \norm{x(0)}$ is independent of $\norm{x(0)}$. So, the advantage is that the `average' network usage remains the same over large portions of the state space.
\end{remark}

\section{Simulation Results}
\label{sec:sim_results}

In this section, the proposed decentralized asynchronous event-triggered sensing mechanism is illustrated with two examples. The first is a linear system and the second a nonlinear system.

\subsection{Linear System Example}

We first present the mechanism for a linearized model of a batch reactor, \cite{walsh2001}. The plant and the controller are given by \eqref{eqn:lin_sys}-\eqref{eqn:lin_con} with
\begin{align*}
&A = \begin{bmatrix} 1.38 & -0.20 &6.71 & -5.67\\ -0.58 & -4.29 & 0 & 0.67\\ 1.06 & 4.27 & -6.65 & 5.89\\ 0.04 & 4.27 & 1.34 & -2.10\end{bmatrix}\\
&B = \begin{bmatrix} 0 & 0\\ 5.67 & 0\\ 1.13 & -3.14\\ 1.13 & 0 \end{bmatrix}\\
&K = - \begin{bmatrix} 0.1006 & -0.2469 & -0.0952 & -0.2447 \\ 1.4099 & -0.1966 & 0.0139 & 0.0823 \end{bmatrix}
\end{align*}
which places the eigenvalues of the matrix $(A+BK)$ at around $\{ -2.98 + 1.19i, -2.98 - 1.19i, -3.89, -3.62 \}$. The matrix $Q$ was chosen as the identity matrix. The system matrices and $Q$ have been chosen to be the same as in \cite{mazo2011async}. Lastly, the controller parameters were chosen as $[ \theta_1, \theta_2, \theta_3, \theta_4 ] = [0.6, 0.17, 0.08, 0.15]$ and $\sigma = 0.95$. For the simulations presented here, the initial condition of the plant was selected as $x(0) = [4, 7, -4, 3]^T$ and the initial sampled data that the controller used was $x_s(0) = [4.1, 7.2, -4.5, 2]^T$. The zeroth sampling instant was chosen as $t_0^{x_i} = - T_i$ for sensor $i$. This is to ensure sampling at $t = 0$ if the local triggering condition was satisfied. Finally the simulation time was chosen as $T_{sim} = 10\text{s}$.

Figures \ref{fig:br_dec_ET_V} and \ref{fig:br_dec_ET_Vdot} show the evolution of the Lyapunov function and its derivative along the trajectory of the closed loop system, respectively. Figure \ref{fig:br_dec_ET_intersamp} shows the time evolution of the inter-transmission times for each sensor. The frequency distribution of the inter-transmission times is another useful metric to understand the closed loop event-triggered system. Thus, given a time interval of interest $[0, \mathcal{T}_{\mathrm{INT}}]$ consider
\begin{align*}
\mathcal{N}^{x_i}( \mathcal{T, T}_{\mathrm{INT}} ) = \Big\{ j \in \mathbb{N}_0 : \ &t^{x_i}_{j+1} \in [0, \mathcal{T}_{\mathrm{INT}}] \quad \mathrm{and} \\
 &( t^{x_i}_{j+1} - t^{x_i}_j ) \in [0, \mathcal{T}] \Big\}
\end{align*}
where $\mathbb{N}_0 = \{0, 1, 2, \ldots \}$ is the set of natural numbers. Hence, the cumulative distribution of the inter-transmission times during $[0, \mathcal{T}_{\mathrm{INT}}]$ is given as
\begin{equation}
\mathcal{D}^{x_i}( \mathcal{T, T}_{\mathrm{INT}} ) = \frac{\# \mathcal{N}^{x_i}( \mathcal{T, T}_{\mathrm{INT}} )} {\# \mathcal{N}^{x_i}( \mathcal{T}_{\mathrm{INT}}, \mathcal{T}_{\mathrm{INT}} )}
\end{equation}
where $\#$ denotes the cardinality of a set. Figure \ref{fig:br_dec_ET_intersamp_dist} shows the cumulative frequency distribution of the inter-transmission times, $\mathcal{D}^{x_i}( \mathcal{T}, T_{sim} )$, for each sensor. The cumulative frequency distribution of the inter-transmission times is a measure of the performance of the event-triggers. A distribution that rises sharply to $100\%$ indicates that event-trigger is not much better than a time-trigger. Thus, slower the rise of the cumulative distribution curves, greater is the justification for using the event-trigger instead of a time-trigger. Note that an implication of the result on scale invariance, Theorem \ref{thm:lin_scale_inv}, is that the cumulative frequency distribution, such as in Figure \ref{fig:br_dec_ET_intersamp_dist} depends only on the ``phase" of the initial condition and not on its magnitude.

\begin{figure*}[!htb]
\centering
\subfloat[]{\label{fig:br_dec_ET_V}\includegraphics[width=0.25\textwidth]{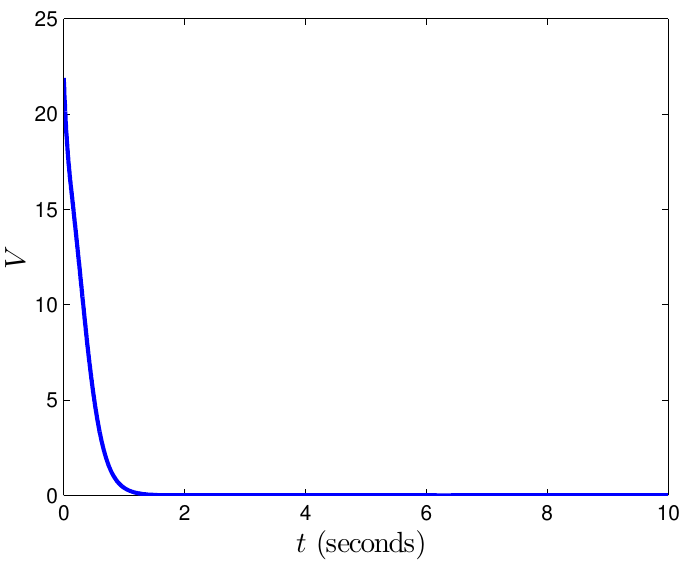}}
\subfloat[]{\label{fig:br_dec_ET_Vdot}\includegraphics[width=0.25\textwidth]{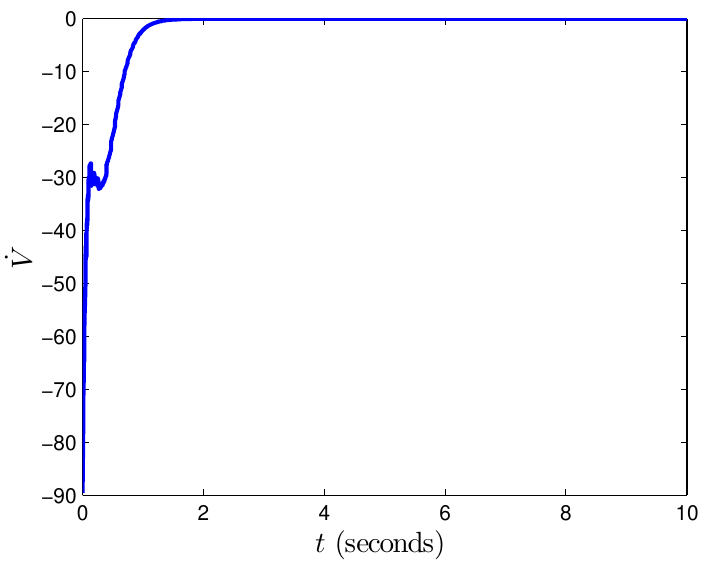}}
\subfloat[]{\label{fig:br_dec_ET_intersamp}\includegraphics[width=0.25\textwidth]{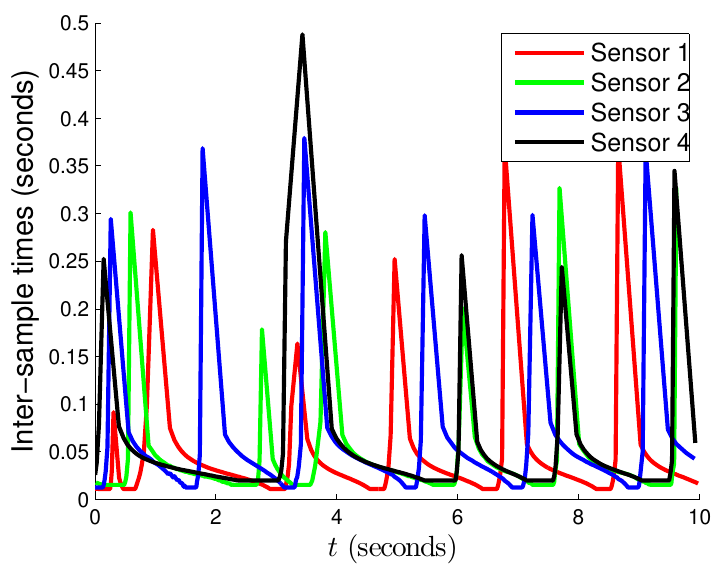}}
\subfloat[]{\label{fig:br_dec_ET_intersamp_dist}\includegraphics[width=0.25\textwidth]{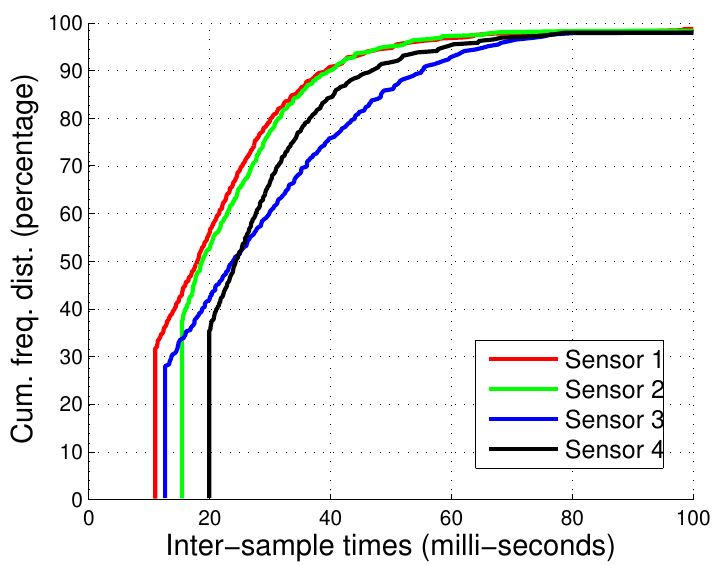}}
\caption{Batch reactor example: evolution of the (a) Lyapunov function, (b) time derivative of Lyapunov function, along the trajectories of the closed loop system. (c) Sensor inter-transmission times (d) cumulative frequency distribution of the sensor inter-transmission times, $\mathcal{D}^{x_i}( \mathcal{T}, T_{sim} )$, where $T_{sim} = 10s$ is the simulation time.}
\end{figure*}

The minimum thresholds for the inter-transmission times $T_i$ for the example can be computed as in \eqref{eqn:Ti_lin} and have been obtained as $[ T_1, T_2, T_3, T_4 ] = [11, 15.4, 12.6, 19.9]\text{ms}$, which are also the minimum inter-transmission times in the simulations presented here. These numbers are a few orders of magnitude higher and an order higher than the guaranteed minimum inter-transmission times and the observed minimum inter-transmission times in \cite{mazo2011async, mazo2012}. The average inter-transmission times obtained in the presented simulations were $[ \bar{T}_1,\bar{T}_2,\bar{T}_3,\bar{T}_4 ] = [ 24.9, 27.7, 34.5, 34.2 ]\text{ms}$, which are about an order of magnitude lower than those reported in \cite{mazo2011async, mazo2012}. A possible explanation for this phenomenon is that in \cite{mazo2011async, mazo2012}, the average inter-transmission times depends quite critically on the evolution of the threshold $\eta$. Although the controller gain matrix $K$ and the matrix $Q$ have been chosen to be the same, by inspection of the plots in \cite{mazo2011async, mazo2012}, it appears that the rate of decay of the Lyapunov function $V$ is roughly about half of that in our simulations. However, we would like to point out that our average inter-transmission times are of the same order as in \cite{mazo2012arxiv} by the same authors. In any case, for LTI systems, our proposed method does not require communication from the controller to sensors to achieve global asymptotic stability. Lastly, as a measure of the usefulness of the event-triggering mechanism compared to a purely time-triggered mechanism, $T_i / \bar{T}_i$ was computed for each $i$ and were obtained as $[ T_1 / \bar{T}_1, T_2 / \bar{T}_2, T_3 / \bar{T}_3, T_4 / \bar{T}_1 ] = [0.44, 0.55, 0.36, 0.58]$. The lower these numbers are, the better it is.

\subsection{Nonlinear System Example}

The general result for nonlinear systems is illustrated through simulations of the following second order nonlinear system.
\begin{align}
&\dot{x} = f(x,x_e) =
\begin{bmatrix}
f_1(x,x_e)\\
f_2(x,x_e)
\end{bmatrix} =
Ax + \begin{bmatrix} 0\\ x_1^3 \end{bmatrix} + Bu \label{eqn:nlin_spring}\\
&\text{where } A = \begin{bmatrix}
0 & 1\\
0 & -1
\end{bmatrix}, \quad B = \begin{bmatrix} 0\\ 1 \end{bmatrix} \notag
\end{align}
where $x = [x_1, x_2]^T$ is a vector in $\mathbb{R}^2$ and the sampled data controller (in terms of the measurement error) is given as
\begin{align}
u = k(x+x_e) = K (x + x_e) - (x_1 + x_{1,e})^3
\end{align}
where $K = [k_1, k_2]$ is a $1 \times 2$ row vector such that $\bar{A} = (A+BK)$ is Hurwitz. Then, the closed-loop system with event-triggered control can be written as
\begin{align}
\dot{x} &= \bar{A}x + BK x_e + \begin{bmatrix} 0\\ x_1^3 - (x_1 + x_{1,e})^3 \end{bmatrix} \notag\\
&= \bar{A}x + \begin{bmatrix} 0\\ h_1 + h_2 \end{bmatrix}
\label{eqn:CL_nlin_spring}
\end{align}
where
\begin{align}
h_1 &= - \Big( x_{1,e}^3 + 3 x_1 x_{1,e}^2 + ( 3 x_1^2 - k_1 ) x_{1,e} \Big) \label{eqn:h1}\\
h_2 &= k_2 x_{2,e} \label{eqn:h2}
\end{align}

Now, consider the quadratic Lyapunov function $V = x^T P x$ where $P$ is a symmetric positive definite matrix that satisfies the Lyapunov equation $P\bar{A} + \bar{A}^T P = - Q$, with $Q$ a symmetric positive definite matrix. Let $p_m$ and $p_M$ be the smallest and largest eigenvalues of the matrix $P$. Since $P$ is a symmetric positive definite matrix, $p_m$ and $p_M$ are each positive real numbers. Further,
\begin{equation*}
\alpha_1(\norm{x}) \triangleq p_m \norm{x}^2 \leq V(x) \leq p_M \norm{x}^2 \triangleq \alpha_2(\norm{x}), \quad \forall x \in \mathbb{R}^2
\end{equation*}
The time derivative of $V$ along the trajectories of the closed loop system \eqref{eqn:CL_nlin_spring} can be shown to satisfy
\begin{align*}
\dot{V} &= -x^T Q x + 2 x^T PB (h_1 + h_2)\\
&\leq -(1-\sigma) Q_m \norm{x}^2 + \norm{x} \big( \norm{2PB (h_1 + h_2)} - \sigma Q_m \norm{x} \big)
\end{align*}
where $Q_m$ is the smallest eigenvalue of the symmetric positive definite matrix $Q$ and $\sigma$ is a parameter satisfying $0 < \sigma < 1$.

Suppose that the desired region of attraction be $S(c)$, for some non-negative $c$ (see \eqref{eqn:S_def} for the definition of $S(c)$). Let $\mu_1$ be the maximum value of $x_1$ on the sub-level set $S(c)$. Then, we let
\begin{align*}
&h_1^c = \norm{x_{1,e}}^3 + 3 \mu_1 \norm{x_{1,e}}^2 +  \max_{\norm{x_1} \leq \mu_1} \{3 x_1^2 - k_1\}\norm{x_{1,e}} \notag\\
&\gamma_1(\norm{x_{1,e}}) \triangleq \frac{\norm{2PB} h_1^c}{\sigma \theta_1 Q_m}, \quad \gamma_2(\norm{x_{2,e}}) \triangleq \frac{\norm{2PBk_2} \norm{x_{2,e}}}{\sigma \theta_2 Q_m}
\end{align*}
where $\theta_1$ and $\theta_2$ are positive constants such that $\theta_1 + \theta_2 = 1$. It is clear that Assumption \eqref{A:ISS} is satisfied and we have
\begin{align*}
\dot{V} &\leq -(1-\sigma) Q_m \norm{x}^2, \quad \text{if } \gamma_i \norm{x_{i,e}} \leq \norm{x}, \quad i \in \{1, 2\}
\end{align*}

Now, $\mu \triangleq \alpha_1^{-1}(c) = \sqrt{c/p_m}$ is the maximum value of $\norm{x}$ on the set $S(c)$. Hence, $M_1(c)$ in \eqref{eqn:M_def} has to be defined for the set on which $\norm{x_{i,e}} \leq R_1 \triangleq \gamma_1^{-1}(\mu)$. Thus, we have that
\begin{align}
&\frac{1}{M_1(c)} = \frac{\norm{2PB}}{\sigma \theta_1 Q_m} \Big( R_1^2 + 3 \mu_1 R_1 + \max_{\norm{x_1} \leq \mu_1} \{3 x_1^2 - k_1\} \Big) \notag \\
&\frac{1}{M_2(c)} = \frac{\norm{2PBk_2}}{\sigma \theta_2 Q_m} \label{eqn:M12_exmp}
\end{align}
Now, only $T_i$ for each $i$ needs to be determined. To this end, the closed loop system dynamics \eqref{eqn:CL_nlin_spring} are bounded as in \eqref{eqn:lip} and \eqref{eqn:lip_i}.
\begin{align*}
| f_1(x,x_e) | &\leq L_1 \norm{x} + D_1 \norm{x_e}\\
| f_2(x,x_e) | &\leq L_2 \norm{x} + D_{2,\mu} \norm{x_e}, \quad \forall x \text{ s.t. } \norm{x} \leq \mu
\end{align*}
Comparing with \eqref{eqn:CL_nlin_spring} the following can be arrived at.
\begin{align*}
&L_1 = \norm{ r_1(\bar{A}) }, \quad D_1 = 0, \quad L_2 = \norm{ r_2(\bar{A}) }\\
&D_{2,\mu} = \sqrt{ \Big( R_1^2 + 3 \mu_1 R_1 + \max_{\norm{x_1} \leq \mu_1} \{3 x_1^2 - k_1\} \Big)^2 + k_2^2 }
\end{align*}

In the example simulation results presented here, the following gains and parameters were used.
\begin{align}
&K = - \begin{bmatrix} 5 & 3 \end{bmatrix}, \quad
Q = \begin{bmatrix}
1 & 0\\
0 & 1
\end{bmatrix}, \quad \theta_1 = 0.9, \quad \theta_2 = 0.1 \notag\\
&\sigma = 0.9, \quad c = 10, \quad \mu_1 = \mu \notag\\
&x(0) = [2.8, -2.6]^T, \quad x_s(0) = [2.9, -2.7]^T \label{eqn:nlin_sim_par}
\end{align}

Notice that $M_2(c)$ is a constant independent of $c$. That is why $\theta_2$ has been chosen much smaller than $\theta_1$. The parameter $\mu_1$ has been chosen to be equal to $\mu$. To be consistent with asynchronous transmissions, the initial value of $x_s(0)$ has been chosen to be different from $x(0)$. The simulation time was chosen as $T_{sim} = 10\text{s}$.

For the chosen parameters and the initial conditions, the initial value of the Lyapunov function is $V(0) = 8.574$. Thus the initial state of the system is well within the region of attraction, given by $S(c) = S(10)$. The event-trigger parameters were obtained as $[w_1, w_2] = [0.0045, 0.0832]$ and $[T_1, T_2] = [4, 3.4]\text{ms}$, which were also the minimum inter-transmission times. The average inter-transmission times of the sensors for the duration of the simulated time were obtained as $[\bar{T}_1, \bar{T}_2] = [4.2, 26.2]\text{ms}$. Thus for sensor 1, the average inter-transmission interval is only marginally better than the minimum. The number of transmissions by sensors 1 and 2 were $2366$ and $382$, respectively.

Figures \ref{fig:nlin_dec_ET_V} and \ref{fig:nlin_dec_ET_Vdot} show the evolution of the Lyapunov function and its derivative along the trajectories of the closed loop system, respectively. Figures \ref{fig:nlin_dec_ET_intersamp} and \ref{fig:nlin_dec_ET_intersamp_dist} show the inter-transmission times and the cumulative frequency distribution of the inter-transmission times, $\mathcal{D}^{x_i}( \mathcal{T}, T_{sim} )$, for each of the sensor. The sharp rise of the cumulative distribution curve for Sensor 1 clearly indicates that the event-triggered transmission is nearly equivalent to time-triggered transmission. On the other hand, the slow rise of the cumulative distribution curve of Sensor 2 demonstrates the usefulness of event-triggering in its case.

\begin{figure*}[!htb]
\centering
\subfloat[]{\label{fig:nlin_dec_ET_V}\includegraphics[width=0.25\textwidth]{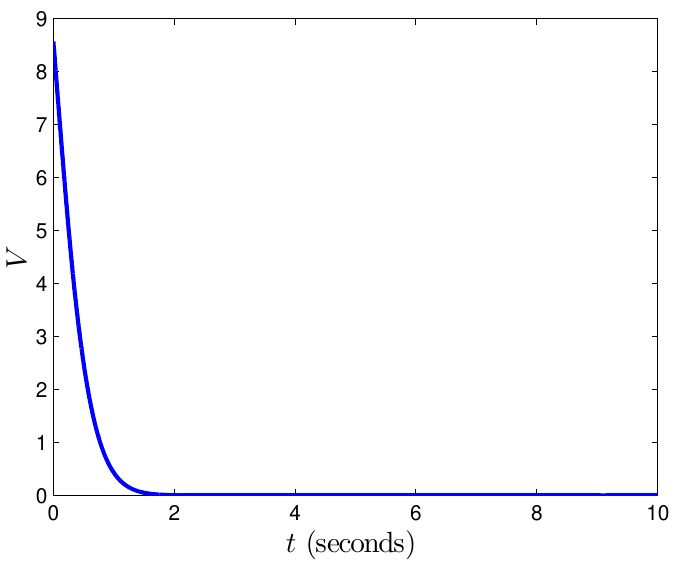}}
\subfloat[]{\label{fig:nlin_dec_ET_Vdot}\includegraphics[width=0.25\textwidth]{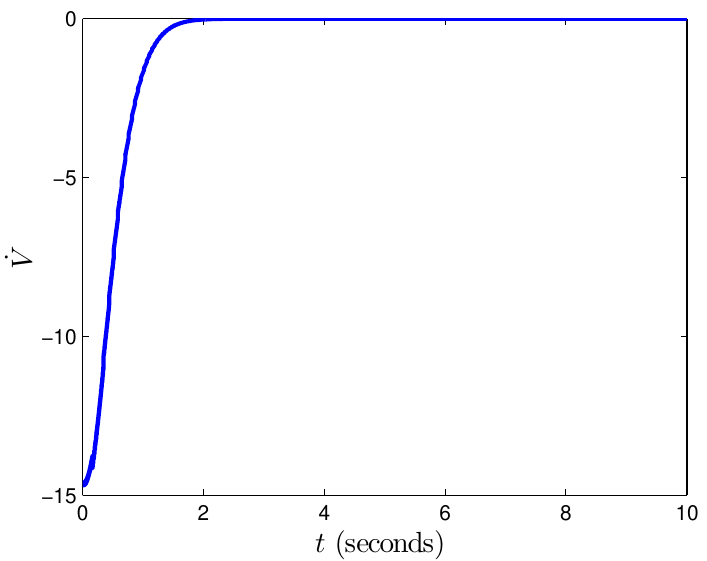}}
\subfloat[]{\label{fig:nlin_dec_ET_intersamp}\includegraphics[width=0.25\textwidth]{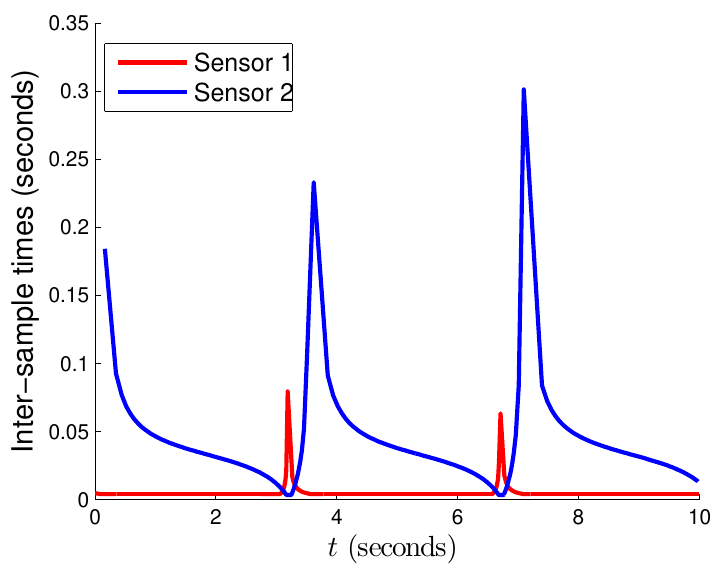}}
\subfloat[]{\label{fig:nlin_dec_ET_intersamp_dist}\includegraphics[width=0.25\textwidth]{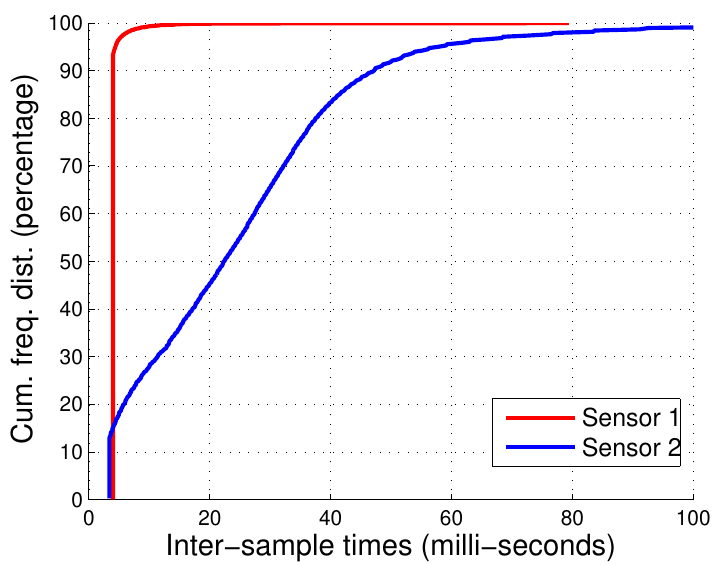}}
\caption{Nonlinear system example: evolution of the (a) Lyapunov function, (b) time derivative of Lyapunov function, along the trajectories of the closed loop system. (c) Sensor inter-transmission times (d) cumulative frequency distribution of the sensor inter-transmission times, $\mathcal{D}^{x_i}( \mathcal{T}, T_{sim} )$, where $T_{sim} = 10s$ is the simulation time.}
\end{figure*}

Simulations were also performed for the case when the central controller intermittently sends updates to the parameters of the sensor event-triggers, as in Theorem \ref{thm:dec_wT}. For the simulation results presented here, the controller gains, parameters and the initial conditions have been chosen the same as in \eqref{eqn:nlin_sim_par}. Additionally, the parameters in \eqref{eqn:ts_calV} were chosen as $\mathcal{T} = 0.5$ and $\rho = 0.5$. The initial condition $\mathcal{V}_s(0) = c = 10$ was chosen.

To obtain the upper bound on the Lyapunov function, $\mathcal{V}(x_s)$ in \eqref{eqn:calV}, the following procedure was adopted. From the event-triggers \eqref{eqn:ET_dec_sensors_bidir}, we have for each $ i \in \{1, \ldots, n\}$ that
\begin{equation*}
\norm{x_{i,s}(t) - x_i(t)} = \norm{x_{i,e}(t)} \leq w_i(t) \norm{x}(t), \ \forall t \geq 0
\end{equation*}
from which we obtain (ignoring the time arguments)
\begin{align*}
&\sum_{i=1}^n \norm{x_{i,s} - x_i}^2 \leq W^2 \sum_{i=1}^n \norm{x_i}^2, \quad \text{where } W = \sqrt{\sum_{i=1}^n w_i^2}\\
&\implies (1-W^2) \sum_{i=1}^n \norm{x_i}^2 - 2 \sum_{i=1}^n \norm{x_{i,s}} \norm{x_i} + \sum_{i=1}^n \norm{x_{i,s}}^2 \leq 0
\end{align*}
which is the equation of an $n$-sphere. Thus, the system state is in the $n$-sphere given by
\begin{align}
&\norm{x - x_c} \leq R \label{eqn:x_sphere}\\
&\text{where }\,\, x_c = \frac{1}{1 - W^2} x_s, \quad R = \frac{W}{1 - W^2} \norm{x_s} \label{eqn:xc_R}
\end{align}
Obviously, for these equations to make sense, $W^2$ has to be strictly less than $1$. 
Notice from \eqref{eqn:M12_exmp},
\begin{align*}
&M_1(c) \leq \frac{\sigma \theta_1 Q_m}{\norm{2PB} (-k1)} \leq 0.4493, \ \forall c \geq 0 \notag \\
&M_2(c) = \frac{\sigma \theta_2 Q_m}{\norm{2PBk_2}} \leq 0.0832, \ \forall c \geq 0
\end{align*}
As a consequence $W^2(t) \leq 0.2088 < 1$ for all $t \geq 0$.

Next from \eqref{eqn:x_sphere}, we know that $\norm{x} \leq \norm{x_c}+R$ and hence that $V(x) \leq \alpha_2(\norm{x_c}+R)$. However, this may be conservative and a better estimate may be obtained by maximizing $V(x)$ on the set given by \eqref{eqn:x_sphere}. In fact, on this set, $V(x)$ is maximized on the boundary of the $n$-sphere. This is because if the maximum does not occur on the boundary and instead occurs only in the interior of the $n$-sphere \eqref{eqn:x_sphere}, then the maximizing sub-level set, $S_M$, of $V$ lies strictly and completely in the interior of the $n$-sphere, which means $S_M$ is not the smallest sub-level set of $V$ that contains the complete $n$-sphere. Thus, an upper bound on the value of $V(x)$ is provided by
\begin{equation}
\mathcal{V}(x_s) \geq \max \{V(x) : \norm{x - x_c} = R \}
\end{equation}
That is, for the $2$ dimensional system in the example \eqref{eqn:nlin_spring}, $\mathcal{V}(x_s)$ is the maximum value of $V$ along a circle. $\mathcal{V}$ was then found in MATLAB by maximization of $V$ on the circle, which was parametrized by a single angle variable varying on the closed interval $[0, 2 \pi]$. Finally, $w_i(t) = M_i(\mathcal{V}_s(t))$ was used for each $i$ and all $t$.

In this case, the number of transmissions by Sensor 1 were much lower at $198$ while those by Sensor 2 were $322$. Notice that $w_2 = M_2(c)$ is a constant, independent of the value of $c$. Thus, we see that the reduction in the number of transmissions by Sensor 2 is only marginal while that of Sensor 1 is huge. The average inter-transmission times of the sensors for the duration of the simulated time were obtained as $[\bar{T}_1, \bar{T}_2] = [50.5, 31.1]\text{ms}$. The minimum inter-transmission times were observed as $4.2\text{ms}$ and $9\text{ms}$ for Sensors 1 and 2, respectively. The number of times the parameters of the sensor event-triggers were updated was $16$.

The evolution of the Lyapunov function and its derivative along the trajectories of the closed loop system were very similar to that in Figures \ref{fig:nlin_dec_ET_V} and \ref{fig:nlin_dec_ET_Vdot}, respectively. Hence, they have not been presented here again. Figures \ref{fig:nlin_dec_mulET_intersamp} and \ref{fig:nlin_dec_mulET_intersamp_dist} show the inter-transmission times and the cumulative frequency distribution of the inter-transmission times, $\mathcal{D}^{x_i}( \mathcal{T}, T_{sim} )$, for each of the sensor. These two plots clearly show the usefulness of the event-triggered transmissions. Figure \ref{fig:nlin_dec_mulET_w} shows the evolution of the $w_i$ parameters of the event-triggers at each of the sensors. As mentioned earlier, $w_2$ is independent of $c$ and hence is a constant. The evolution of $w_1$ shows that it is a non-decreasing function of time. Finally, Figure \ref{fig:nlin_dec_mulET_T} shows the evolution of the $T_i$ parameters of the event-triggers at the sensors (for clarity $T_2$ has been scaled by $20$ times). Although, $T_1$ evolves in a non-decreasing manner, the same is not the case with $T_2$.

\begin{figure*}[!htb]
\centering
\subfloat[]{\label{fig:nlin_dec_mulET_intersamp}\includegraphics[width=0.25\textwidth]{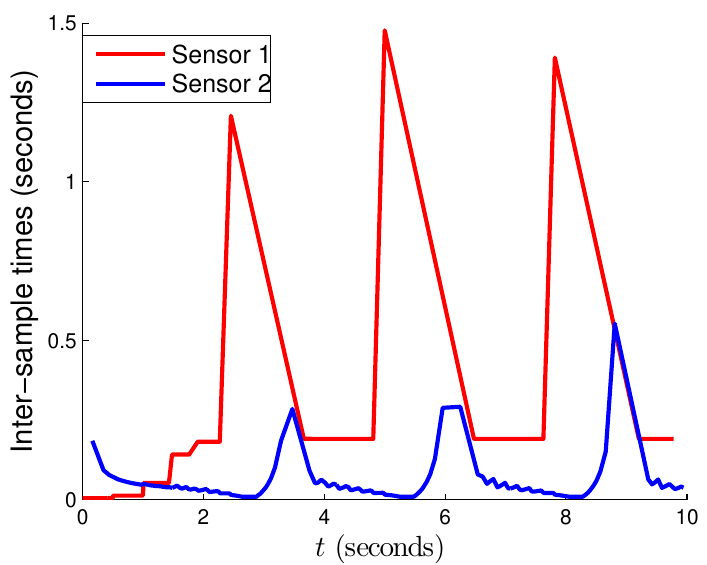}}
\subfloat[]{\label{fig:nlin_dec_mulET_intersamp_dist}\includegraphics[width=0.25\textwidth]{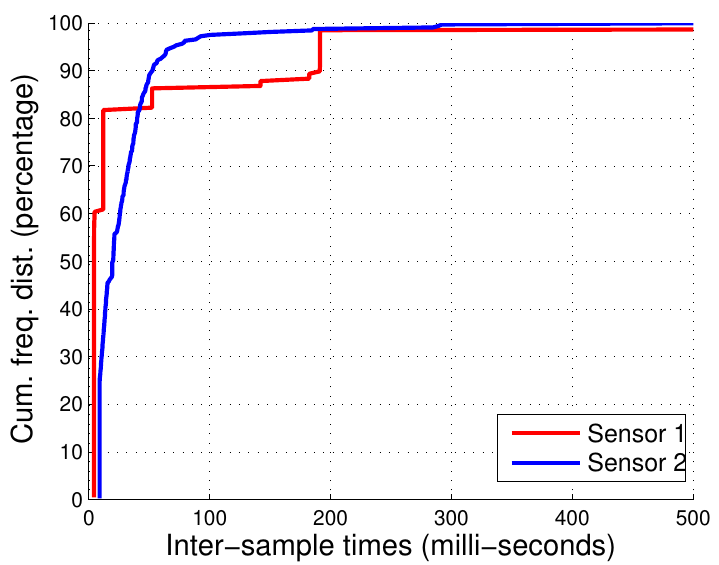}}
\subfloat[]{\label{fig:nlin_dec_mulET_w}\includegraphics[width=0.25\textwidth]{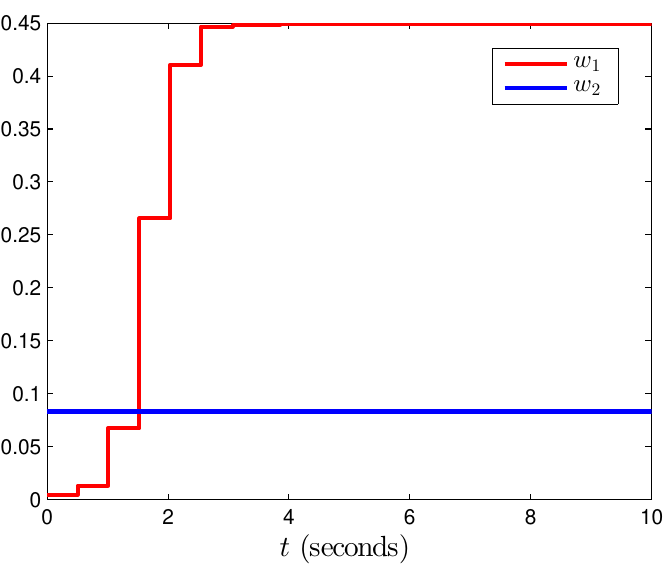}}
\subfloat[]{\label{fig:nlin_dec_mulET_T}\includegraphics[width=0.25\textwidth]{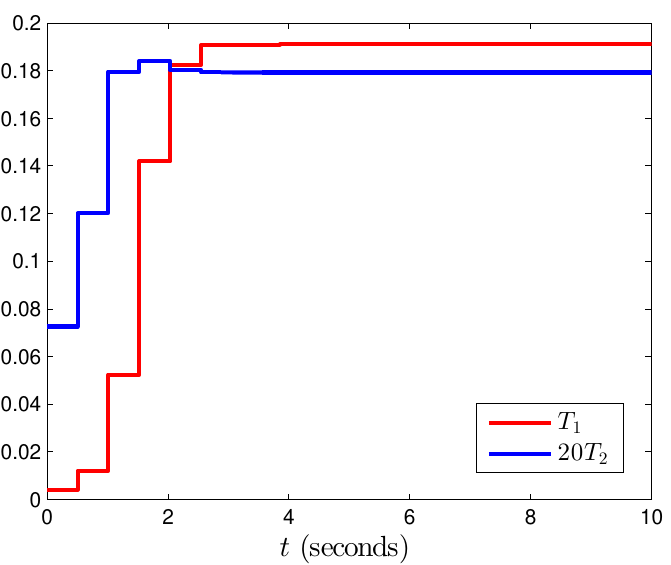}}
\caption{Nonlinear system example with event-triggered communication from the controller to the sensor event-triggers: (a) Sensor inter-transmission times (b) cumulative frequency distribution of the sensor inter-transmission times, $\mathcal{D}^{x_i}( \mathcal{T}, T_{sim} )$, where $T_{sim} = 10s$ is the simulation time. Evolution of (c) $w_i$, (d) $T_i$ parameters of the sensor event-triggers.}
\end{figure*}

\section{Conclusions}
\label{sec:conc}

In this paper, we have developed a method for designing decentralized event-triggers for control of nonlinear systems. The architecture of the systems considered in this paper included full state feedback, a central controller and distributed sensors not co-located with the central controller. The aim was to develop event-triggers for determining the time instants of transmission from the sensors to the central controller. The proposed design ensures that the event-triggers at each sensor depend only on locally available information, thus allowing for asynchronous transmissions from the sensors to the central controller. Further, the design aimed at completely eliminating (or drastically reducing) the need for the sensors to listen to other sensors and/or the controller.

The proposed design was shown to guarantee a positive lower bound for inter-transmission times of each sensor (and of the controller in one of the special cases). The origin of the closed loop system is also guaranteed to be asymptotically stable with an arbitrary, but priorly fixed, region of attraction. In the special case of linear systems, the region of attraction was shown to be global with absolutely no need for the sensors to listen. Finally, the proposed design method was illustrated through simulations of a linear and a nonlinear example.

In the system architecture considered in this paper, although the control input to the plant is updated intermittently, it is not exactly event-triggered. In fact, in all the results the inter-transmission times of each sensor individually have been shown to have a positive lower bound. And the time interval between receptions at the central controller from two different sensors can be arbitrarily close to zero. Since the control input to the plant is updated each time the controller receives some information, no positive lower bound can be guaranteed for the inter-update times of the controller. However, it is not very difficult to additionally incorporate event-triggering (with guaranteed positive minimum inter-update times) or an explicit threshold on inter-update times of the control, as in \cite{mazo2011}.

Next, although the transmissions of sensors have been designed to be asynchronous, the communication from the central controller to the sensors in Section \ref{sec:2way_com} have been assumed to be synchronous. In future, we aim to allow these communications also to be asynchronous. Although time delays have not been considered explicitly, they may be handled as in most event-triggered control literature (see \cite{tabuada2007} for example). It is worthwhile to investigate more sophisticated triggers for updating the parameters $w_i$ and $T_i$ (Section \ref{sec:2way_com}) as is a thorough study and quantification of sensor listening effort. Finally, our results were short of mathematically demonstrating an improvement in the inter-transmission times for the scheme of Section \ref{sec:2way_com} compared to that of Section \ref{sec:1way_com}. We believe that a promising approach to the quantification of any improvement is through analytical characterization of the frequency distribution of the inter-transmission times, $\mathcal{D}^{x_i}( \mathcal{T}, T_{INT} )$.


\bibliographystyle{IEEEtran}
\bibliography{IEEEabrv,control_refs}

\end{document}